\documentclass{article}
\usepackage[a4paper,
            bindingoffset=0.2in,
            left=1in,
            right=1in,
            top=1in,
            bottom=1in,
            footskip=.25in]{geometry}

\usepackage[T1]{fontenc}
\usepackage[utf8]{inputenc}
\usepackage[dvipsnames]{xcolor}
\usepackage{algpseudocode}
\usepackage{hyperref}

\usepackage{amsthm}
\usepackage{algorithm}
\usepackage{amsmath}
\usepackage{amssymb}
\usepackage{mathtools}
\usepackage{physics}
\usepackage{dsfont}
\usepackage{accents}
\usepackage{xfrac}
\usepackage{mathrsfs}
\usepackage{enumerate}
\usepackage{MnSymbol}
\usepackage{xcolor}
\usepackage[ligature, inference]{semantic}
\usepackage{ebproof}
\usepackage{graphicx}
\usepackage{stmaryrd}
\usepackage{xspace}
\usepackage{multirow}
\usepackage{array}
\usepackage{subcaption}
\usepackage{import}
\usepackage{tikz}
\usepackage{paralist}
\usepackage{tabularx}
\usepackage{adjustbox}
\usepackage{subcaption}
\usepackage{rotating}
\usepackage{stackengine}
\usepackage{marvosym}
\usepackage{multicol}
\usepackage{thm-restate}
\declaretheorem[name=Theorem]{theorem}
\declaretheorem[name=Lemma]{lemma}
\declaretheorem[name=Definition]{definition}

\declaretheorem[name=Example]{example}

\usetikzlibrary{positioning,fit}
\usetikzlibrary{quantikz2}
\usetikzlibrary{patterns}

\tikzset{
mystyle/.style ={preaction={fill, white}, pattern =north west lines, pattern color=black}
}
\tikzset{
mystyle2/.style ={preaction={fill, white},pattern =crosshatch dots, pattern color=black}
}

\usepackage{wrapfig}

\newcommand{\gatecolor}{gray!30}

\definecolor{procedures}{RGB}{202,82,213}
\definecolor{commands}{RGB}{9,10,184}

\newcommand{\codecomm}[1]{\text{\textcolor{gray}{/* #1 */}}}

\newcommand{\No}{\textnormal{Anc}}
\newcommand{\ia}{\mathrm{i}}
\newcommand{\ja}{\mathrm{j}}
\newcommand{\ka}{\mathrm{k}}
\newcommand{\xa}{\mathrm{x}}
\newcommand{\sa}{\mathrm{s}}
\newcommand{\bb}{\mathrm{b}}
\newcommand{\ea}{\mathrm{e}}
\newcommand{\da}{\mathrm{d}}
\newcommand{\LN}{\mathcal{L}(\mathbb{N})}
\newcommand{\N}{\mathbb{N}}
\newcommand{\B}{\mathbb{B}}
\newcommand{\Z}{\mathbb{Z}}
\newcommand{\q}{\mathrm{q}}

\renewcommand{\op}{\mathrm{op}}
\newcommand{\skp}{\tb{ skip};}
\newcommand{\tb}[1]{\text{\color{commands}\textnormal{\textbf{#1}}}}
\newcommand{\tbt}[1]{\text{\textnormal{\textbf{#1}}}}
\newcommand{\pn}[1]{\text{\color{procedures}\textnormal{\texttt{#1}}}}
\newcommand{\el}[1]{\textnormal{\small[$#1$]}}
\newcommand{\proc}{\pn{proc}}

\newcommand{\decl}{\tb{decl }}
\newcommand{\call}{\tb{call }}
\renewcommand{\ket}[1]{\ensuremath{\left| #1\right\rangle}}

\newcommand{\kpsi}{\ket{\psi}}
\newcommand{\kpsip}{\ket{\psi'}}
\newcommand{\kpsipp}{\ket{\psi''}}
\newcommand{\kpsik}{\ket{\psi_k}}
\newcommand{\qcase}[3]{\tb{qcase } #1 \tb{ of }\{0\to #2, 1\to #3\}}

\newcommand{\cif}[3]{\tb{if } #1 \tb{ then } #2 \tb{ else }#3}

\newcommand{\ST}{\mathrm{S}}
\newcommand{\CST}{\mathrm{Cst}}
\newcommand{\D}{\mathrm{D}}
\newcommand{\PR}{\ensuremath{\mathrm{P}}}

\newcommand{\U}{\mathrm{U}}
\newcommand{\asg}{{\ \textnormal{\texttt{\textasteriskcentered =}}\ }}

\newcommand{\contL}{\ensuremath{l_\text{M}}} 

\newcommand{\stdv}{\kpsi,A,l}

\newcommand{\width}{\textnormal{width}}
\newcommand{\fbqp}{\textnormal{\textsc{fbqp}}}

\newcommand{\wf}{\textnormal{\textsc{wf}}}
\newcommand{\wi}{\textnormal{\textsc{width}}}
\newcommand{\unif}{\textnormal{\textsc{basic}}}

\newcommand{\pfoq}{\textnormal{\textsc{pfoq}}}
\newcommand{\pbp}{\textnormal{\textsc{pbp}}}

\newcommand{\size}[1]{| #1 |}
\newcommand{\sem}[1]{\llbracket #1 \rrbracket}
\newcommand{\topbot}{\diamond}
\newcommand{\semto}[1]{\stackrel{#1}{\longrightarrow}}
\newcommand{\compile}{\textnormal{\textbf{compile}}\xspace}

\newcommand{\optimize}{\ensuremath{\textbf{optimize}}}

\newcommand{\level}{\textnormal{Time}}

\newcommand{\compgroup}[2]{\gategroup[#1,steps=#2,style={dashed,fill = gray!0, inner sep=0pt},background]{}}

\newcommand{\ttt}[1]{\text{\scriptsize\texttt{\textcolor{gray}{#1}}}}

\renewcommand\fbox{\fcolorbox{gray}{white}}

\title{Branch Sequentialization in Quantum Polytime} 

\author{Emmanuel Hainry \and Romain Péchoux \and Mário Silva}

\date{%
    Université de Lorraine, CNRS, Inria, LORIA, F-54000 Nancy, France\\[2ex]%
    \today\\[2ex]
    {\small \texttt{\{hainry,pechoux,mmachado\}@loria.fr}}
}

\begin{document}

\maketitle

\begin{abstract}
Quantum algorithms leverage the use of quantumly-controlled data in order to achieve computational advantage. 
This implies that the programs use constructs depending on quantum data and not just classical data such as measurement outcomes. Current compilation strategies for quantum control flow involve compiling the branches of a quantum conditional, either in-depth or in-width, which in general leads to circuits of exponential size. This problem is coined as the branch sequentialization problem. We introduce and study a compilation technique for avoiding branch sequentialization on a language that is sound and complete for quantum polynomial time, thus, improving on existing polynomial-size-preserving compilation techniques.
\end{abstract}

\captionsetup{justification=centering}

\section{Introduction}
\label{s:intro}

\subsection{Motivation} 


Quantum computing is an emerging paradigm of computation, where quantum physical phenomena such as entanglement and superposition are used to obtain an advantage over classical computation. A testament to the richness of the field is the variety of computational models: quantum Turing machines~\cite{BernsteinVazirani}, quantum circuits~\cite{Yao1993,NielsenChuang}, measurement-based quantum computation~\cite{B09,DK06}, linear optical circuits~\cite{KLM01}, among others. 

In recent decades, a lot of effort has been put into developing high-level quantum programming languages that allow programmers to abstract away from the technicalities of these low-level models~\cite{S04,AG05}. Toward that end, several verification techniques such as type systems~\cite{FKS20} or categorical approaches for reasoning about program semantics~\cite{B06,H19} have been studied and developed to ensure the physical reality of compiled programs, for example, by ensuring that they satisfy the properties of quantum mechanics such as the no-cloning theorem~\cite{A09} and unitarity~\cite{DCGMV19}.

An important line of research in this area involves checking  polytime termination of quantum programs~\cite{DALLAGO2010377,Yamakami20, HPS23}, by showing that any program can be simulated by a polytime quantum Turing machine (QTM). As a consequence of Yao's Theorem~\cite{Yao1993}, demonstrating that polynomial-time QTMs are computationally equivalent to uniform and poly-sized quantum circuit families, such programs can be instantiated by a uniform family of quantum circuits of polynomial size, i.e., with a polynomial number of qubits and gates. However, direct compilation techniques often fail to ensure this property, due to the use of quantum branching, where the flow in a conditional is determined by the state of a qubit.

As an illustrative example, consider a \emph{quantum case}, of the shape $$\tbt{qcase }\q_1,\dots,\q_k\tbt{ of }\{i\to \ST_i\}_{i\in\{0,1\}^k},$$  which executes the superposition of the unitary transformations $U_i$ implemented by statements $\ST_i$, depending on the state $\ket{i}$ of the control qubits $\q_1,\dots,\q_k$. This statement can be simulated by a QTM  that, depending on $k$ symbols in its tape, performs the instructions of different QTMs $M_i$ in superposition. Consequently, its runtime is the maximum runtime among all $M_i$ (\cite[Branching Lemma]{BernsteinVazirani}). Instantiating this $\tbt{qcase}$ statement on quantum circuits can be done sequentially (in-depth) or in parallel (in-width). 
An in-depth encoding consists in applying each controlled-$U_i$ gate sequentially, thus implementing the unitary transformation $\sum_{i\in\{0,1\}^k} \ketbra{i}{i} \otimes U_i$. The depth of the resulting circuit is the sum of the depths of all $U_i$, which is exponential in the number $k$ of control qubits.
An alternative strategy, inspired by QRAM, allows for the execution of each $U_i$ in parallel in the circuit, using controlled swaps to perform routing of qubit addresses in linear depth on the number of control qubits. Since in this case the unitaries are parallelized, the circuit depth is the maximum among the depths of $U_i$, but this strategy results in circuits where the width scales as the sum of the complexities of $U_i$. Consequently,  this parallelization requires a number of ancillas that grows exponentially on the number of control qubits. In summary, any implementation, that is agnostic about the structure of each $U_i$, should incur such an exponential cost, as each \tbt{qcase} corresponds to the preparation of an arbitrary quantum state~\cite{ZLY22,GKW24}.

Automatic implementations of  quantum conditionals therefore produce circuits with a size that scales with the sum (rather than the maximum) of the complexity of each branch. This problem, coined \emph{branch sequentialization} in~\cite{Yuan2022}, leads, in many cases, to an exponential blow-up in circuit complexity, forcing the programmer to rewrite and optimize their code in order to make use of the program structure and improve the complexity bound. A challenging issue is, therefore, to develop quantum programming languages that avoid branch sequentialization by ensuring the correct circuit complexity for quantum control statements while providing full use of quantum control to the programmer~\cite{Yuan2022,yuan2024quantum}.

\subsection{Contribution} 
This paper solves the above issue by introducing a programming language, called $\pbp$ for Polynomially Branching Program, with quantum case and  first-order recursive procedures, on which compilation of quantum conditional does not lead to an exponential blow-up in circuit size. Our main contributions are as follows:
\begin{itemize}
\item We introduce a compilation strategy \compile{} (Figures~\ref{fig:compile} and~\ref{fig:optimize-steps}) into quantum circuits, show its soundness (Theorem~\ref{thm:soundcompile}), and show that it avoids branch sequentialization on $\pbp$ (Theorem~\ref{thm:new-compilation}), a language with restrictions on recursive calls to procedures. Toward that end, we formally define the time complexity $\level_\PR : \mathbb{N} \to \mathbb{N}$ of a program $\PR$ (Definition~\ref{def:timecomp}) as a map from the number $n$ of qubits to the maximal number of procedures called during a program execution on the Hilbert space $\mathbb{C}^{2^n}$. The time complexity of  a $ \tbt{qcase}$ statement is precisely the maximum of the time complexity of its branches and we show that branch sequentialization is avoided as the circuits generated by $\compile$ have size bounded asymptotically by the time complexity of the compiled program (Theorem~\ref{thm:new-compilation}).
\item A natural question concerns the expressive power of language $\pbp$. We also solve this issue by  showing that $\pbp$ is sound and complete for quantum polynomial time (Theorem~\ref{thm:soundness-and-completeness}). Consequently, any polytime quantum algorithm can be encoded by a $\pbp$ program. We illustrate the methodology through well-known examples such as the program \texttt{QFT} realizing the quantum Fourier transform (Example~\ref{fig:QFTu}).
\item We also discuss asymptotic bounds for a strict extension of $\pbp$, where the reduction in the input sorted set passed to procedure calls can be arbitrary. We show that this extension only increases the circuit complexity of the no-branch-sequentialization case by a linear factor (Theorem~\ref{thm:non-basic-compilation}).
\item  We also show that \compile{} strictly improves on existing compilation strategies on first-order quantum programs with quantum case~\cite{HPS23} by a polynomial speedup of arbitrarily large degree (Table~\ref{table:problems}).
\end{itemize}

\subsection{A bird's-eye view of our compilation strategy} 
We will now consider a simple illustrating example where branch sequentialization can exponentially worsen the size complexity of the compiled circuit, and show how our compilation strategy can preserve a polynomial bound.

The program $\texttt{PAIRS}$ defined in Figure~\ref{fig:BS-example1} consists in a simple call to procedure \texttt{pairs} (line 10) on a list $\bar{\q}$ of unique qubits. Let $\size{\bar{\q}}$ be the number of qubits in $\bar{\q}$.
By language design, the procedure \texttt{pairs} immediately terminates whenever $\size{\bar{\q}}=0$. The procedure enters a quantum case (line 3) when $\size{\bar{\q}} \geq 2$, otherwise it applies a $\text{NOT}$ gate to the single qubit (line 9). On line 3,  the program will branch depending on the state $\bar{\q}\el{1,2}$ of the first two qubits in $\bar{\q}$. Out of all four cases (lines 4-7), \texttt{pairs} only performs an operation when the first two qubits are in state $\ket{00}$ or $\ket{11}$, in which case it performs a recursive call on $\bar{\q}\ominus \el{1,2}$, the list obtained by removing the first and second qubits of $\bar{\q}$. 

\begin{wrapfigure}{r}{0.46\textwidth}
\centering
\vspace{-5mm}
\begingroup
\addtolength{\jot}{0mm}
\begin{align*}
\ttt{1}&\quad \tb{decl }\pn{pairs}(\bar{\q})\{\\
\ttt{2}& \qquad \tb{if }|\bar{\q}| \geq 2 \tb{ then }\\
\ttt{3}&\quad\qquad \tb{qcase }\bar{\q}\el{1,2}\tb{ of }\{\\
\ttt{4}&\qquad\qquad 00\to \tb{call }\pn{pairs}(\bar{\q}\ominus \el{1,2});\\
\ttt{5}&\qquad\qquad 01\to \tb{skip};\\
\ttt{6}&\qquad\qquad 10\to \tb{skip};\\
\ttt{7}&\qquad\qquad 11\to \tb{call }\pn{pairs}(\bar{\q}\ominus \el{1,2});\\
\ttt{8}&\qquad\qquad \}\\
\ttt{9}&\qquad \tb{else }\bar{\q}\el{1}\asg \text{NOT};\}\\
\ttt{10}&\quad ::\ \tb{call }\pn{pairs}(\bar{\q});
\end{align*}
\endgroup
\vspace{-8mm}
\captionsetup{justification=centering}
  \caption{Program $\texttt{PAIRS}$.}
  \label{fig:BS-example1}
\end{wrapfigure} 

Given an input state $\ket{xy}$, with $x\in\{0,1\}^*$ and $y\in\{0,1\}$,  \texttt{pairs} will apply a $\text{NOT}$ gate to the last qubit $y$ if and only if $x$ is a string in the language $(00\mid 11)^\ast$. Since \texttt{pairs} performs at most one call per branch, and consumes 2 qubits from its input while doing so, its time complexity is in $O(\size{\bar{\q}})$, when taking the \tbt{qcase} complexity as the maximum of each branch.

We now discuss the circuit compilation of \texttt{pairs}, when $\size{\bar{\q}} \geq 2$. In Figure~\ref{fig:strategies}, we give three strategies~\textbf{(a)}, \textbf{(b)}, and \textbf{(c)}, exemplifying different approaches. Strategies~\textbf{(a)} and \textbf{(b)} are automatic compilation strategies, which ignore the inner structure of the program -- the fact that the body of the two non-skip branches of the \textbf{qcase} are  identical and therefore encode the same unitary transformation.
Figure~\ref{fig:strategies}\textbf{(a)} represents an in-depth strategy, whereas Figure~\ref{fig:strategies}\textbf{(b)} implements an in-width strategy, making use of different registers $r_{00}$ and $r_{11}$, both initialized to $\ket{0^{\size{\bar{\q}}-2}}$ to perform each branch in parallel and then recombine the results in the same register. These two implementations require two recursive calls to \texttt{pairs} and, consequently, their size complexity is the sum of the complexities of the two branches. These two strategies are performing branch sequentialization and, in both cases, the compiled circuit has size exponential in $\size{\bar{\q}}$. 
In contrast, Figure~\ref{fig:strategies}\textbf{(c)} avoids this exponential blow-up, making use of the fact that the two branches are identical. This strategy is able to \emph{merge} the recursive calls into a single procedure call with the use of one ancilla. 
This is precisely the type of strategy used by our compilation algorithm $\compile$ (Section~\ref{s:compilation}). Although this example is relatively simple due to similar recursive calls, $\compile$ allows to deal with more involved situations, by merging recursive calls on different parameters.


\begin{figure}
\centering
\scalebox{0.65}{
\begin{tikzpicture}
\node[] (A) at (0,0) {
\begin{quantikz}[wire types = {b}]
\lstick{$\phantom{}\bar{\q}$} & \gate{\texttt{pairs}(\bar{\q})} & \rstick{\phantom{$\bar{\q}$}}
\end{quantikz}};

\node[] (B) at (-8,-2.6) {
\scalebox{0.8}{
\begin{quantikz}[wire types = {q,q,q}]
\lstick{$\bar{\q}\el{1}$} & & \octrl{1} & \ctrl{1} &\\
\lstick{$\bar{\q}\el{2}$} & & \octrl{1} & \ctrl{1} &\\
\lstick{$\bar{\q}\ominus\el{1,2}$} &\qwbundle{} & \gate{\texttt{pairs}(\bar{\q} \ominus \el{1,2})} & \gate{\texttt{pairs}(\bar{\q} \ominus \el{1,2})} &
\end{quantikz}}};

\node[] (D) at (0,-3.5) {
\scalebox{0.8}{
\begin{quantikz}[wire types = {q,q,q,q}, column sep = 3mm, row sep = 4mm]
\lstick{$\bar{\q}\el{1}$}  & & \octrl{1} & &\ctrl{1} & & &&\ctrl{1} & &\octrl{1} &\\
\lstick{$\bar{\q}\el{2}$} & &\octrl{2} & &\ctrl{4} & & & & \ctrl{4} & &\octrl{2} &\\
\lstick{$\bar{\q}\ominus\el{1,2}$} & \qwbundle{} & &\swap{2} & & \swap{4} &  & \swap{4} & &\swap{2} & &\\
\lstick{\ket{0}} & &\targ{} & \ctrl{0}& & & \ctrl{1} & & & \ctrl{0} &\targ{} &\\[-2mm]
\lstick{$r_{00}$} &\qwbundle{} & &\targX{} & & &\gate{\texttt{pairs}(\bar{\q} \ominus \el{1,2})} & & &\targX{} & &\\[-1mm]
\lstick{$\ket{0}$} & & & & \targ{} & \ctrl{0} & \ctrl{1} &\ctrl{0}& \targ{} && &\\[-2mm]
\lstick{$r_{11}$} &\qwbundle{} & & & &\targX{} &\gate{\texttt{pairs}(\bar{\q} \ominus \el{1,2})} & \targX{} & & & &
\end{quantikz}}
};

\node[] (C) at (8,-3) {
\scalebox{0.8}{
\begin{quantikz}[wire types = {q,q,q,q}]
\lstick{$\bar{\q}\el{1}$}  & \octrl{1} & \ctrl{1} & & \ctrl{1} & \octrl{1} &\\
\lstick{$\bar{\q}\el{2}$} & \octrl{1} & \ctrl{1} & & \ctrl{1} & \octrl{1} &\\
\lstick{$\ket{0}$} & \targ{} & \targ{} & \ctrl{1} &  \targ{} & \targ{} &\\
\wireoverride{n}& \lstick{$\bar{\q}\ominus\el{1,2}$}\wireoverride{n}  & \qwbundle{} &  \gate{\texttt{pairs}(\bar{\q}\ominus\el{1,2})}  &  && 
\end{quantikz}}};

\draw[->] (A) to [bend right=20] node[yshift = 5mm,xshift=0mm] {\textbf{(a)} in-depth} (B)  ;
\draw[->] (A) to [bend left=20] node[yshift = 5mm, xshift=1mm] {\textbf{(c)} merging} (C);
\draw[->] (A) to [bend right=0] node[yshift = 0mm, xshift= 12mm] {\textbf{(b)} in-width} (D);
 \end{tikzpicture}}
 \captionsetup{justification=centering}
 \caption{Compilation strategies.}
 \label{fig:strategies}
\end{figure}

\subsection{Related work}
Resource optimization in low-level models of quantum computation is a well-studied subject. Given a specific quantum circuit, it is possible to reduce its number of gates (or at least its number of non-Clifford gates) via techniques such as gate substitution and graph rewriting~\cite{NRSCM18,MDMN08,KvdW20,dBBW20}.
The study of resource optimization in the \emph{asymptotic} scenario is a relatively young research area as it involves reasoning about families of circuits and program structure. Different implicit characterizations of quantum complexity classes have been developed using a lambda-calculus~\cite{DALLAGO2010377}, function algebras~\cite{Yamakami20,yamakami2022} and a first-order programming language~\cite{HPS23}. The last of these provided a compilation strategy into quantum circuits with bounds on the circuit size based on the syntactic restrictions placed on the programs.

Compilation strategies for the quantum control statement (also called quantum switch case or quantum multiplexer) have been studied, for instance, in state preparation~\cite{GS01}, appearing in quantum machine learning~\cite{Harrow_2009} and Hamiltonian simulation~\cite{CW12}. These techniques either focus on optimizing number of qubits (circuit width)~\cite{ZLY22} or circuit depth~\cite{STYYZ23}, but in all cases correspond to circuits of exponential size. In order to improve on these bounds, one must restrict the set of programs to those that admit an efficient implementation, which can be deduced from the program structure. Optimized compilation techniques in that scenario can then be judged on expressivity and completeness: how easily can one write a program while ensuring the syntactical restrictions? Do these restrictions encompass all efficient programs? The field of implicit computational complexity is particularly well-suited to answer these questions~\cite{CdL24,CdL25}.

\section{First-Order Programming Language with Quantum Case}\label{s:language}

In this section, we introduce a programming language with quantum control and first-order recursive procedures.
After introducing its syntax and semantics, we introduce a restriction $\pbp$ (for Polynomially Branching Programs) on which branch sequentialization can be avoided.


\begin{figure}[h]
\hrulefill
$$ \begin{array}{lllll}
\text{(Integers)} &\ia,\ja,\ka & ::= & \xa \ | \ n \ | \ \ia \pm n \ | \ \size{\sa}\\
\text{(Booleans)} &\bb & ::= &  \ia \geq \ia \ | \  \cdots \ | \ \bb \wedge \bb \ | \ \cdots \\
\text{(Qubits)} & \q & ::= & \sa\el{\ia}\\
\text{(Sorted sets)} &\sa & ::= & \bar{\q} \ | \ 
\sa\ominus\el{\ia}\\
\text{(Statements)} & \ST &::= & \skp\ |\  \q \asg \U^f(\ja); \ | \ \ST\ \ST \ |\ \tb{if }\bb\tb{ then }\ST\tb{ else } \ST    \\
& & &\ |\  \qcase{\q}{\ST}{\ST}\ |\ \call \proc(\sa); & \\
\text{(Procedure declarations)} & \D &::= & \ \varepsilon \ |\  \decl \proc(\bar{\q})\{\ST\},\,\D \\
\text{(Programs)} & \PR(\bar{\q}) &::= &\,\D::\ST 
\end{array}$$
\hrulefill
\captionsetup{justification=centering}
\caption{Syntax of programs.}
\label{fig:syntax}
\end{figure}

\subsection{Syntax of  Programs}
The language includes four basic datatypes for expressions, whose corresponding expressions are described in Figure~\ref{fig:syntax}.
\begin{inparaenum}[(i)]
\item \emph{Integers}: Integer expressions are variables $\xa$, constant $n \in \mathbb{N}$, an addition by a constants $\ia \pm n$, or the size  $\size{\sa}$ of a sorted set $\sa$.
\item \emph{Booleans}: Such expressions $\bb$ are defined in a standard way.
\item \emph{Qubit}: qubits expressions are of the shape $\sa\el{\ia}$ which denotes the $\ia$-th qubit in sorted set $\sa$. \item \emph{Sorted sets}: lists of unique (i.e., non-duplicable) qubit pointers. Sorted-set expressions $\sa$ are either variables $\bar{\q}$ 
 or $\sa\ominus\el{\ia}$, the sorted set $\sa$ where the $\ia$-th element has been removed. 
\end{inparaenum}
Let $\sa\el{\ia_1,\ldots,\ia_n}$ be a shorthand for $\sa\el{\ia_1},\ldots,\sa\el{\ia_n}$. We also define syntactic sugar for pointing to the $n$-th \emph{last} qubit in a sorted set, by defining for any $n\geq1$, $\bar{\q}\el{-n}\triangleq \bar{\q}\el{|\bar{\q}|-n+1}$.

A program $\PR\triangleq \D :: \ST$ is defined in Figure~\ref{fig:syntax} as a list of (possibly recursive) procedure declarations $\D$, followed by a program statement $\ST$. 

Let Procedures be an enumerable set of procedure names. We write $\proc\in\PR$ to denote that $\proc$ appears in $\PR$.
 Each procedure of name $\proc \in \PR$ is defined by a (unique) procedure declaration $\tb{decl }\proc(\bar{\q})\{\ST\} \in \D$, which takes a sorted set $\bar{\q}$ as input parameter and executes the \emph{procedure statement} $\ST$. 
 We sometimes write $\ST^{\proc}$ to explicitly refer to the procedure statement $\ST$ of $\proc$.

Statements include the no-op instruction, unitary application, sequences, conditional, quantum case, and procedures calls. For the sake of universality~\cite{barenco1995elementary}, in a unitary application $\q \asg \U^f(\ja);$, $\U^f(\ja)$ is a unitary transformation that can take an integer $\ja$ and a polynomial-time approximable~\cite{Adleman1997} total function $f \in \mathbb{Z} \to [0,2\pi)$ as optional arguments.
The $f$ and $\ia$ can be omitted when they are not useful, as in a \text{NOT} gate.

 The quantum conditional $\tb{qcase }\sa\el{\ia}\tb{ of }\{0\to \ST_0,1\to\ST_1\}$ allows branching by executing statements $\ST_0$ and $\ST_1$ in superposition according to the state of qubit $\sa[\ia]$. The procedure call $\tb{call }\proc(\sa);$ runs procedure $\proc$ with \emph{sorted set} expression $\sa$. The quantum conditional can be extended to multiple qubits in a standard way as used in Figure~\ref{fig:BS-example1}.

We also make use of some syntactic sugar to describe statements encoding constant-time quantum operations. For instance, the $\text{CNOT}$, $\text{SWAP}$, 
as well as a controlled-phase shift gate are defined by:
\begin{align*}
\text{CNOT}(\q_1,\q_2)\ \triangleq&\ \tb{qcase }\q_1\tb{ of }\{0\to\tb{skip};,\,1\to\q_2 \asg \text{NOT};\}\\
\text{SWAP}(\q_1,\q_2)\ \triangleq& \ \text{CNOT}(\q_1,\q_2) \ \text{CNOT}(\q_2,\q_1)\ \text{CNOT}(\q_1,\q_2)\\
\textnormal{CPHASE}(\q_1,\,\q_2,\,\ia)\triangleq& \tb{ qcase }\q_1 \tb{ of }\{0\to\tb{skip};,1\to \q_2\asg\text{Ph}^{\lambda x.\pi/2^{x-1}}(\ia);\}
\end{align*}

Given a program $\PR\triangleq \D::\ST$, the call relation $\to_{\PR}\ \subseteq \textnormal{Procedures} \times \textnormal{Procedures}$ is defined for any two procedures $\pn{proc$_1$},$ $\pn{proc$_2$} \in \PR$ as $\pn{proc$_1$} \to_{\PR} \pn{proc$_2$}$ whenever $\pn{proc$_2$} \in \ST^{\pn{proc$_1$}}$. The relation $\succeq_{\PR}$ is then the transitive closure of $\to_{\PR}$, and $\sim_{\PR}$ denotes the equivalence relation where $\pn{proc$_1$} \sim_{\PR} \pn{proc$_2$} $ if $\pn{proc$_1$}\succeq_{\PR} \pn{proc$_2$}$ and $\pn{proc$_2$} \succeq_{\PR} \pn{proc$_1$}$ both hold.
A procedure $\proc $ is \emph{recursive} whenever $\proc\sim_{\PR} \proc$ holds.
The strict order $\succ_{\PR}$ is defined as $\pn{proc$_1$}\succ_{\PR} \pn{proc$_2$} $ if $\pn{proc$_1$}\succeq_{\PR} \pn{proc$_2$}$ and  $\pn{proc$_1$} \not\sim_{\PR} \pn{proc$_2$}$ both hold.

\subsection{Semantics of Programs} 

Let $\mathbb{B}$ denote the set of Booleans and $\LN$ denote the set of lists of natural numbers, $[]$ being the empty list. We interpret each basic type as follows:
\begin{align*}
\sem{\text{Integers}} &\triangleq \mathbb{Z} &
\sem{\text{Booleans}} &\triangleq \mathbb{B}  &\sem{\text{Qubits}} &\triangleq \mathbb{N} &
\sem{\text{Sorted Sets}} &\triangleq \LN
\end{align*}
Qubits are interpreted as integers (pointers) and sorted sets are interpreted as lists of pointers.
For each basic type $\tau$, the semantics of expressions is described standardly as a function 
$$\Downarrow_{\sem{\tau}} \ \in \tau \times \LN \to \sem{\tau}$$

$(\ea,l)\Downarrow_{\sem{\tau}} v$ holds when expression $\ea$ of type $\tau$ evaluates to the value $v\in\sem{\tau}$ under the context $l \in \LN$. $l$ is simply the sorted set of qubit pointers into consideration when evaluating $\ea$. For instance, we have that $(\bar{\q}\el{2},[1,4,5])\Downarrow_\mathbb{N} 4$ (the second qubit has index $4$), $(\bar{\q}\ominus \el{4},[1,4,5])\Downarrow_{\LN} []$ ($[]$ is used for errors on type $\LN$), $(\bar{\q}\el{4},[1,4,5])\Downarrow_\mathbb{N} 0$ (index $0$ encodes error on type $\mathbb{N}$), and $(\bar{\q}\ominus\el{3},[1,4,5])\Downarrow_{\LN} [1,4]$ (the third qubit has been removed). 

\begin{figure}[t]
\hrulefill
\begin{center}
$
\scalebox{0.8}{
 \begin{prooftree}
 \hypo{\phantom{\Downarrow_{\mathbb{N}}}}
 \infer1[
 ]{(\tb{skip};,\stdv)\semto{0} (\top,\stdv)}
 \end{prooftree}
 }
\quad 
\scalebox{0.8}{
 \begin{prooftree}
 \hypo{(\q,l)\Downarrow_{\mathbb{N}} n \notin A}
 \infer1[
 ]{(\q\asg \U^f(\ja);,\stdv)\semto{0} (\bot,\stdv)}
 \end{prooftree}
}
$
\\[0.2cm]
$
\scalebox{0.8}{
 \begin{prooftree}
 \hypo{({\q},l)\Downarrow_{\mathbb{N}} n \in A}
 \hypo{(\ja, l)\Downarrow_{\mathbb{Z}} k}
 \infer2[
 ]{(\q\asg \U^f(\ja);,\stdv)\semto{0} (  \top, I_{2^{n-1}}\otimes \sem{\U}(f)(k)\otimes I_{2^{l({\kpsi})-n}} \kpsi ,A, l )}
 \end{prooftree}
 }
$
\\[0.2cm]
$
\scalebox{0.8}{
  \begin{prooftree}
  \hypo{(\ST_1,\stdv)\semto{m_1} (\top,\kpsip,A,l)}
   \hypo{(\ST_2,\kpsip,A,l)\semto{m_2} (\topbot,\kpsipp,A,l)}
 \infer2[
 ]{(\ST_1\ \ST_2,\stdv)\semto{m_1+m_2} ( \topbot,\kpsipp,A,l)}
 \end{prooftree}
 }
$
\\[0.2cm]
$
\scalebox{0.8}{
  \begin{prooftree}
  \hypo{( \ST_1,\stdv)\semto{m} (\bot,\stdv)}
  \infer1[
  ]{( \ST_1\ \ST_2,\stdv)\semto{m}(\bot,\stdv)}
 \end{prooftree}
}
$
\quad
$
\scalebox{0.8}{
\begin{prooftree}
  \hypo{(\bb,l) \Downarrow_{\B} b}
  \hypo{(\ST_b,\stdv)\semto{m_b} (\topbot,\kpsip,A,l)}
  \hypo{\topbot \in \{\top,\bot\}}
  \infer3[
  ]{(\cif{\bb}{\ST_{\tb{true}}}{\ST_{\tb{false}}} ,\stdv)\semto{m_b}(\topbot ,\kpsip,A,l)}
  \end{prooftree}
}
$
\\[0.2cm]
$
\scalebox{0.8}{
  \begin{prooftree}
  \hypo{(\q,l) \Downarrow_{\mathbb{N}} n \in A}
  \hypo{\forall k \in \{0,1\}, \ (\ST_{k},\kpsi,A\backslash \{{n}\},l)\semto{m_k} (\top,\kpsik,A\backslash \{{n}\},l)}
  \infer2[
  ]{ ( \qcase{\q}{\ST_0}{\ST_1},\stdv)\semto{\max_{k \in \{0,1\}} m_k}(\top,\sum_{k} \ket{k}_{n}\! \bra{k}_{n} \kpsik,A,l)}
  \end{prooftree}
  }
$
\\[0.2cm]
$
\scalebox{0.8}{
  \begin{prooftree}
  \hypo{(\q,l) \Downarrow_{\mathbb{N}} n \in A}
  \hypo{\forall k \in \{0,1\}, \ (\ST_{k},\kpsi,A\backslash \{{n}\},l)\semto{m_k} (\topbot_k,\kpsik,A\backslash \{{n}\},l)}
\hypo{\bot\in\{\topbot_0, \topbot_1\}}
  \infer3[
  ]{( \qcase{\q}{\ST_0}{\ST_1},\stdv)\semto{\max_{k \in \{0,1\}} m_k}(\bot,\stdv)}
  \end{prooftree}
}
$
\\[0.2cm]
$
\scalebox{0.8}{
  \begin{prooftree}
  \hypo{({\q},l) \Downarrow_{\mathbb{N}} n \notin A}
  \infer1[
  ]{( \qcase{\q}{\ST_0}{\ST_1},\stdv)\semto{0}(\bot,\stdv)}
  \end{prooftree}
}
\quad
\scalebox{0.8}{
\begin{prooftree}
 \hypo{(\sa,l) \Downarrow_{\mathcal{L}(\mathbb{N})} [\,]}
 \infer1[
 ]{(\call \proc (\sa);,\stdv)\semto{1} (\top,\stdv)}
\end{prooftree}
}
$
\\[0.2cm]
$
\scalebox{0.8}{
\begin{prooftree}
  \hypo{(\sa,l) \Downarrow_{\mathcal{L}(\mathbb{N})} l' \neq [\,]}
  \hypo{(\ST^{\proc},\kpsi,A, l' )\semto{m} (\topbot,\kpsip,A,l')}
  \hypo{\topbot \in \{\top,\bot\}}
  \infer3[
  ]{(\call \proc(\sa);,\stdv)\semto{m+1} (\topbot,\kpsip,A,l)}
\end{prooftree}
}
$
\end{center}
\hrulefill
\caption{Semantics of statements.}
\label{table:operationalsemantics}
\end{figure}

Let $\mathcal{H}_{2^n} $ denote the Hilbert space of $n$ qubits $\mathbb{C}^{2^n}$, $\mathcal{P}(\mathbb{N})$ denote the powerset of $\mathbb{N}$. 

A \emph{configuration} $c$ over $n$ qubits is of the shape $(\ST,\stdv)$, for some statement $\ST \in \text{Statements}\cup \{\top,\bot\})$, $\top$ and $\bot$ being two special symbols denoting termination and error, respectively, $\kpsi$ is a quantum state in $ \mathcal{H}_{2^n}$, where $A \in \mathcal{P}(\mathbb{N})$ is the set of accessible qubit pointers, and where $l \in \mathcal{L}(\mathbb{N})$ is the list of qubit pointers under consideration.
Let $\text{Conf}_n$, be the set of configurations over $n$ qubits.
The initial configuration in $\text{Conf}_{n}$ of a program $\D::\ST$ on input state $\kpsi \in \mathcal{H}_{2^n}$ is $c_{init}(\kpsi) \triangleq (\ST,\kpsi,\{1,\dots,n\},[1,\dots,n]) $. A final configuration can be defined in the same way as 
$c_{final}(\kpsi) \triangleq (\top,\kpsi,\{1,\dots,n\},[1,\dots,n]) $.

Each unitary transformation $ \U$ of a unitary application $ \bar{\q}\el{\ia} \asg \U^f(\ja);$, comes with a function $\sem{\U}$ assigning a unitary matrix $\sem{\U}(f)(n) \in \mathbb{C}^{2 \times 2}$ to each integer $n$ and  polynomial-time approximable total function $f \in \mathbb{Z} \to [0,2\pi)$. 
For example, the gates of the quantum Fourier transform can be defined by $R_n \triangleq \sem{\textnormal{Ph}}(\lambda x. \pi/2^{x-1})(n)$ with $\sem{\textnormal{Ph}}(f)(n) \triangleq (\begin{smallmatrix} 1 &0 \\ 0 & e^{if(n)} \end{smallmatrix})$. The other basic unary gates are the \text{NOT} and the $R_Y$ gate.



The big-step semantics $\cdot \semto{\cdot} \cdot $ is defined in Figure~\ref{table:operationalsemantics} as a relation in $\bigcup_{n \in \mathbb{N}} \text{Conf}_n \times \mathbb{N} \times \text{Conf}_n$. 
Standard notations from quantum computation are used such as tensor product $\otimes$, $\bra{\psi}$ for the conjugate transpose of $\ket{\psi}$, or given a dimension $m$, $\ket{k}_n \triangleq I_{2^{n-1}}\otimes\ket{k}\otimes I_{2^{m-n}}$ for $k\in\{0,1\}$.
In Figure~\ref{table:operationalsemantics}, the set $A$ of accessible qubits is used to ensure that unitary operations on qubits can be physically implemented. For example, to ensure reversibility, in a quantum branch $\tb{qcase }{\sa\el{\ia}}\tb{ of }\{0\to\ST_0,\,1\to \ST_1\}$, statements $\ST_0$ and $\ST_1$ cannot access $\sa\el{\ia}$.

\begin{definition}\label{def:timecomp}
The \emph{time complexity} $\level_\PR : \mathbb{N} \to \mathbb{N}$ of a program $\PR\triangleq \D::\ST$ is defined by
$$\level_\PR(n) \triangleq \max_{\ket\phi\in\mathcal{H}_{2^n}} \{m \in \mathbb{N} \mid \exists \ket{\phi'}\in \mathcal{H}_{2^n},\ c_{init}(\ket{\phi}) \semto{m}  c_{final}(\ket{\phi'})\}.$$
\end{definition}
Intuitively, when $c \semto{m} c'$ holds, the \emph{time complexity} $m$ is an integer corresponding to the maximum number of procedure calls performed over each (classical and quantum) branch during the evaluation of a configuration $c \in \text{Conf}_n$. 
We write $\sem{\PR}(\kpsi) = \kpsip$, whenever $c_{init}(\kpsi) \semto{m}c_{final}(\kpsip)$ holds for some $m$. If the program $\PR$ terminates on any input (i.e., always reaches a final configuration) then $\sem{\PR}$ is a total function on quantum states.

\begin{example}\label{ex:pairs} For the program $\texttt{PAIRS}$ of Figure~\ref{fig:BS-example1}, $\level_\texttt{PAIRS}(n)=\lfloor \frac{n}{2}\rfloor+1$ since each procedure call removes two qubits until it reaches a sorted set $\bar{\q}$ such that $|\bar{\q}| \leq 1$ (depending on whether $n$ is odd or even) and for both sizes there are no more procedure calls.
\end{example}

\subsection{Polynomial Branching Programs}

\setlength{\tabcolsep}{2pt}
\begin{figure}[t]
\centering
\fbox{
\scalebox{0.9}{
\begin{tabular}{c}
{
\begin{tabular}{rlrlrl}
\ttt{1}&$\quad\tb{decl }\pn{qft}(\bar{\q})\{$ & \ttt{6} & $\quad\tb{decl }\pn{rot}(\bar{\q})\{$ & \ttt{11} & $\quad\tb{decl }\pn{shift}(\bar{\q})\{$\\
\ttt{2}&$\qquad \bar{\q}\el{1}\asg \text{H};$ & \ttt{7} & $\qquad \tb{if }|\bar{\q}|>1 \tb{ then }$ & \ttt{12} & $\qquad \tb{if }|\bar{\q}|>1 \tb{ then }$\\
\ttt{3}&$\qquad \tb{call }\pn{rot}(\bar{\q});$ & \ttt{8} & $\quad\qquad \textnormal{CPHASE}(\bar{\q}\el{-1},\,\bar{\q}\el{1},\,|\bar{\q}|)\ \ $ & \ttt{13} & $\quad\qquad\text{SWAP}(\bar{\q}\el{1},\bar{\q}\el{-1})$\\
\ttt{4}&$\qquad \tb{call }\pn{shift}(\bar{\q});$ &\ttt{9}  & $\quad\qquad \tb{call }\pn{rot}(\bar{\q}\ominus\el{-1});$  & \ttt{14}  & $\quad\qquad\tb{call }\pn{shift}(\bar{\q} \ominus \el{-1});\ \ $ \\
\ttt{5}&$\qquad \tb{call }\pn{qft}(\bar{\q}\ominus\el{-1});\},\ \ $ & \ttt{10} &$\qquad\tb{else skip};\},\ \ $& \ttt{15} &$\qquad\tb{else skip};\}$ \\
\ttt{16} & $\quad::$\\
\ttt{17} & \multicolumn{4}{l}{$\quad\tb{call }\pn{qft}(\bar{\q});$}
\end{tabular}
}
\\
{\begin{scalebox}{0.9}{
\begin{quantikz}[wire types = {q,q,q,q}, row sep = 4mm, column sep = 4mm]
 &\gate{\text{H}} & \gate{\text{R}_4} & \gate{\text{R}_3} &  \gate{\text{R}_2} & \permute{4,1,2,3}\gategroup[4,steps=1,style={draw=gray, inner	sep=-1pt},background]{\pn{shift}} & \gate{\text{H}} & \gate{\text{R}_{3}}&\gate{\text{R}_2} & \permute{3,1,2}\gategroup[3,steps=1,style={draw=gray, inner sep=-1pt},background]{} & \gate{\text{H}}& \gate{\text{R}_2} & \permute{2,1}\gategroup[2,steps=1,style={draw=gray, inner sep=-1pt},background,label style={label position=below,anchor=north,yshift=-0.2cm}]{} &  \gate{\text{H}}  &  \\
 &   & &  & \ctrl{-1} & &  &   & \ctrl{-1} &  &   &    \ctrl{-1} &  & &   \\
 &   &  & \ctrl{-2}  & &    &  &  \ctrl{-2} &  & & &  &  &  &   \\
 & & \ctrl{-3} &  &   &    &  &    & &   & & &  &     &   \end{quantikz}}\end{scalebox}}\\
 \\
\end{tabular}}}
\caption{Program $\texttt{QFT}$ for the quantum Fourier transform.}
\label{fig:QFTu}
\end{figure}

We define three restrictions on the programming language to consider only a strict subset, called $\pbp$ for Polynomial Branching Programs, on which branch sequentialization can be avoided.
First, we define a well-foundedness criterion to consider only terminating programs. A program $\PR$ is \emph{well-founded} if each recursive procedure call removes at least one qubit in its parameter. \wf{} denotes the set of well-founded programs. 
Then, we define a criterion to exclude programs with exponential runtime. Toward that end, the notion of \emph{width of a procedure} $\proc$ in a program $\PR$ is introduced. 

\begin{definition}
Given a program $\PR$, the \emph{width} of a procedure $\pn{proc} \in \PR$, denoted $\width_{\PR}(\pn{proc})$, is defined as $\width_{\PR}(\pn{proc}) \triangleq w^{\pn{proc}}_{\PR}(\ST^{\pn{proc}})$, where $w_{\PR}^{\pn{proc}}(\ST)$  is defined inductively as follows:

\begin{align*}
w_{\PR}^{\pn{proc}}(\tb{skip};)&\triangleq 0\\
w_{\PR}^{\pn{proc}}(\q\emph{\asg} \U^f(\ia);)&\triangleq 0,\\
w_{\PR}^{\pn{proc}}(\ST_1 \ \ST_2)&\triangleq w^{\pn{proc}}_{\PR}(\ST_1)+ w^{\pn{proc}}_{\PR}(\ST_2),\\ 
w_{\PR}^{\pn{proc}}(\tb{if }\bb\tb{ then }\ST_0\tb{ else }\ST_1)&\triangleq \max(w^{\pn{proc}}_{\PR}(\ST_0),w^{\pn{proc}}_{\PR}(\ST_1)),\\
w_{\PR}^{\pn{proc}}(\tb{qcase }{\q}\tb{ of }\{0\to\ST_0,\,1\to\ST_1\})&\triangleq \max(w^{\pn{proc}}_{\PR}(\ST_0),w^{\pn{proc}}_{\PR}(\ST_1)),\\
w_{\PR}^{\pn{proc}}(\tb{call }\pn{proc'}(\sa);)&\triangleq \begin{cases} 1& \text{if }\pn{proc}\sim_{\PR} \pn{proc'},\\ 0 &\text{otherwise}. \end{cases}
\end{align*}
 Let $\wi_{\leq 1}$ be the set of programs with procedures of width at most $1$.
\end{definition}
 Programs of width $1$ are inherently polynomial as they cannot perform an exponential number of procedure calls in sequence. However the total number of calls in superposition may be exponential for such programs by definition of the width for conditional and quantum case.
 
 Finally, for the purpose of our compilation process, we impose further syntactical conditions on programs. We restrict our attention to what we will call \unif{} programs. Let \unif{} denote the set of programs where each call to a procedure is performed either on the variable $\bar{\q}$ or on a unique sorted set $\sa$ that is fixed for the whole program.
 
 \begin{definition}[Polynomially Branching Programs]
The set $\pbp$ of \emph{polynomially branching programs} is defined as  $\pbp{}\triangleq \wf \cap \wi_{\leq 1} \cap \unif{}$.
 \end{definition}
 
Notice that $\pbp$  is strictly included in the programming language of~\cite{HPS23}, that roughly corresponds to $\wf \cap \wi_{\leq 1}$, where procedure calls can take an extra integer parameter.

\begin{example}\label{ex:qft} The \textnormal{\texttt{QFT}} program written in Figure~\ref{fig:QFTu} is in $\pbp{}$. Indeed,  the program $\textnormal{\texttt{QFT}}$ is in $\unif{}$ as calls are only performed on either $\bar{\q}$ or $\bar{\q}\ominus\el{-1}$. It is also in  $\wf{}$ as all recursive calls are performed on parameter $\bar{\q}\ominus\el{-1}$. Finally, it is in $\wi_{\leq 1}$ as $\width_\textnormal{\texttt{QFT}}(\pn{qft})=\width_\textnormal{\texttt{QFT}}(\pn{rot})=\width_\textnormal{\texttt{QFT}}(\pn{shift})=1$.
\end{example}
 
The restriction to programs in $\pbp$ does not impact negatively the expressive power of our study from an extensional viewpoint as, by Theorem~\ref{thm:soundness-and-completeness}, $\pbp$ is sound and complete for the class \fbqp{} of functions in $\{0,1\}^* \to \{0,1\}^*$, i.e., functions that can be approximated with at least $2/3$ probability by a quantum Turing machine running in polynomial time~\cite{BernsteinVazirani,Yamakami20}.

\section{Compilation Strategy} \label{s:compilation}
We now present the compilation algorithm  from $\wf \cap \wi_{\leq 1}$ to quantum circuits based on the merging strategy of Figure~\ref{fig:strategies}\textbf{(c)}. 
Compilation is restricted to programs in $\wf \cap \wi_{\leq 1}$ for two reasons.
The well-foundedness criterion $\wf$ ensures that the compilation always terminates. The restriction to $\wi_{\leq 1}$ prevents exponential blow-up.
Note however that, in Theorem~\ref{thm:new-compilation}, the $\pbp$ restrictions are required to avoid the branch sequentialization problem.

\subsection{Anchoring and Merging}

\begin{figure}[b]
\setlength{\tabcolsep}{1mm}
\hrulefill

\begin{center}
\includegraphics[scale=0.95]{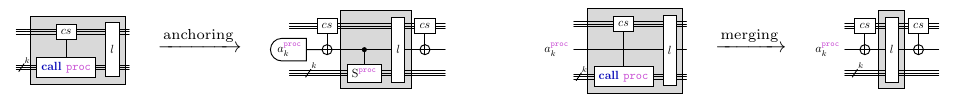}
\end{center}
\hrulefill
\caption{Anchoring and merging.}
\label{fig:example-legend}
\end{figure}

A \emph{control structure} $cs$ is a partial map in $\N\to\{0,1\}$. Intuitively, $cs$ represents a map from qubit pointers to their control values for a controlled gate and will also be used as a shorthand notation in circuits. For $n \in \N$ and $k \in \{0,1\}$, let $cs[n := k]$ be the control structure obtained from $cs$ by setting $cs(n) = k$.  We denote by $dom(cs)$ the domain of the control structure and by $\cdot$, the control structure such that $dom(\cdot)=\emptyset$.

During the compilation of the program statement, every recursive call is handled as follows: if it is the first call with this procedure name and input size, an ancilla is created and its procedure statement is compiled under a control on this ancilla (\emph{anchoring}), otherwise this procedure call has already been compiled and the control structure of the new call is sent into the corresponding ancilla (\emph{merging}).
Anchoring and merging are represented in Figure~\ref{fig:example-legend} as rewrite rules on quantum circuits for
a procedure call $\tb{call }\proc(\sa)$, $cs$ denoting the control structure corresponding to the (possibly nested) quantum cases, where the procedure call occurs.
$a_k^\proc$ represents the ancilla introduced (in case of anchoring) or reused (in case of merging). The integer $k$ refers to the number $\size{\sa}$ of qubits in the procedure call. The grey box materializes the controlled statements left to compile.

 \subsection{The compile Algorithm}
\label{ss:algorithm}
 
 \begin{figure}[t]
\hrulefill

\centering
$
\scalebox{0.7}{
\begin{quantikz}[wire types = {b,b}, row sep = 2mm,column sep = 2mm]
& \gate{cs}\vqw{1}\compgroup{2}{1}&\\
& \gate{\tb{skip};} &
\end{quantikz}}\longrightarrow
\scalebox{0.7}{\begin{quantikz}[wire types = {b}, column sep = 2mm]
  & &
\end{quantikz}} 
$
\qquad 
$
\scalebox{0.7}{
\begin{quantikz}[wire types = {b,q,b}, row sep = {8mm,between origins},column sep = 2.5mm]
 & \gate{cs}\vqw{1}\compgroup{3}{1} &\\
\lstick{$n$} & \gate[2, style ={inner sep = 0mm}]{\bar{\q}\el{n}\asg \U^f(m);} &\\[-2mm]
&  &
\end{quantikz}}\longrightarrow
\scalebox{0.7}{
\begin{quantikz}[wire types = {b,q,b}, row sep = {7mm,between origins},column sep = 2.5mm]
&\gate{cs}\vqw{1}&  \\
\lstick{$n$} & \gate{\sem{\U}(f)(m)} & \\
&  & 
\end{quantikz}}
$
\qquad 
$
\scalebox{0.7}{
\begin{quantikz}[wire types = {b,b}, column sep = 2mm]
& \gate{cs}\vqw{1}\compgroup{2}{1} & \\
& \gate{\ST_1 \ \ST_2} & 
\end{quantikz}
}
\longrightarrow
\scalebox{0.7}{
\begin{quantikz}[wire types = {b,b}, column sep =2mm]
  &\gate{cs}\vqw{1}\compgroup{2}{1} &[3mm] \gate{cs}\vqw{1} \compgroup{2}{1}&\\
& \gate{\ST_1} & \gate{\ST_2} &
\end{quantikz}}
$
\\[0.2cm]
$
\scalebox{0.7}{
\begin{quantikz}[wire types = {b,b}, row sep = 2mm,column sep = 2mm]
 & \gate{cs}\vqw{1}\compgroup{2}{1}&\\
& \gate{\tb{if }b\tb{ then }\ST_1\tb{ else }\ST_0} &
\end{quantikz}}
\longrightarrow
\scalebox{0.7}{
\begin{quantikz}[wire types = {b,b},column sep = 2mm]
 &\gate{cs}\vqw{1}\compgroup{2}{1}& \\
& \gate{\ST_b} & 
\end{quantikz}} 
$
\qquad
$
 \scalebox{0.75}{
\begin{quantikz}[wire types = {b,q,b}, row sep = 2mm,column sep = 2mm]
 & \gate{cs}\vqw{1}\compgroup{3}{1} &\\
\lstick{$n$}& \gate[2, disable auto height]{\begin{matrix}\tb{qcase }\bar{\q}\el{n}\tb{ of }\\
\{0\to\ST_0,\,1\to\ST_1\}\end{matrix}} &\\
& &
\end{quantikz}}
\longrightarrow
\scalebox{0.75}{
\begin{quantikz}[wire types = {b,q,b},row sep =3mm,column sep = 2mm]
 &\gate{cs}\vqw{1}\compgroup{3}{1} &[2mm]&\gate{cs}\vqw{1}\compgroup{3}{1} &\\
\lstick{$n$}& \octrl{1} & &\ctrl{1} &\\
& \gate{\ST_0} & &\gate{\ST_1}   &
\end{quantikz}} 
$
\\[0.2cm]
$
\scalebox{0.65}{\vspace{2mm}
\begin{quantikz}[wire types = {b,b},row sep = 2mm,column sep = 2mm]
&  \gate{cs} \vqw{1}\compgroup{2}{1} & \\
 &\gate{\tb{call }\proc(\sa);} & 
\end{quantikz}}
\longrightarrow\left\{
\scalebox{0.6}{
\begin{quantikz}[wire types = {b,b},column sep = 2mm]
 & \gate{cs} \vqw{1}\compgroup{2}{1} &\\
& \gate{\ST^\proc\{\sa/\bar{\q}\}}& 
\end{quantikz}}
\ \ \text{\footnotesize if $w_\PR^\proc(\ST^\proc)=0,$}\quad
\scalebox{0.6}{
\begin{quantikz}[wire types = {b,b},column sep = 4mm]
 &     \gate{cs} \vqw{1} \gategroup[2,steps=3,style={inner sep=0pt,fill=\gatecolor,draw=black},background]{$\proc$}  & \gategroup[2,steps=1,style={mystyle, inner sep=-2pt},background]{}  & &\\
 & \gate{\text{S}^\proc\{\sa/\bar{\q}\}}&   & &
\end{quantikz}}\quad\text{\footnotesize\text{otherwise.}}\right\}
$

\hrulefill
\caption{Rewrite rules of \compile{}.}
    \label{fig:compile}
\end{figure}
 Let controlled statements be pairs of the shape $(cs,\ST)$, for some control structure $cs$ and some statement $\ST$.
 In the compilation process, controlled statements $(cs,\ST)$ are used to represent a statement $\ST$ that remains to be compiled into a quantum circuit $C$ together with a control structure $cs$, representing the qubits controlling the compiled circuit $C$. 
The compilation algorithm $\compile$ is described by the rewrite rules of Figure~\ref{fig:compile}.  
Generating the circuit corresponding to the program $\PR=\D::\ST$ will consist in running $\compile$ on the controlled statement $(\cdot,\ST)$ for a fixed list of qubits pointers $[1,\ldots,n]$ (hence, an input of fixed size $n$).
We denote by $\compile{}(\PR,n)$ the circuit obtained for program $\PR$ on size $n$.
The algorithm standardly generates the quantum circuit corresponding to a $\wf \cap \wi_{\leq 1}$ program inductively on the statement $\ST$. Indeed, in the rules of Figure~\ref{fig:compile} when $\compile$ is called on a given controlled statement $(cs,\ST)$, an inductive call to $\compile$ is performed on controlled statements whose statements are sub-statements of $\ST$. The two base case are the rules for the skip statement and the unitary application. In these cases, the compilation just outputs the identity circuit and a controlled gate computing the unitary, respectively. The rules for sequence and quantum case perform two inductive calls to $\compile$ on each branch of the statement.
The rule for quantum case is the only rule that directly performs changes on the control structure.
In the particular case of a call to a recursive procedure (i.e., when $w_\PR^\proc(\ST^\proc)>0$), $\compile$  calls the $\optimize$ subroutine to perform anchoring and merging. This call to $\optimize$ is highlighted in Figure~\ref{fig:compile} through the use of a shaded square $\begin{quantikz} \gate[style = {fill=\gatecolor, inner sep=0pt}]{}\end{quantikz}$, which takes the procedure name $\proc$ as superscript. We call this process the \emph{optimization} of procedure $\proc$.

\begin{figure}[t]
\hrulefill

 \begin{minipage}{0.98\textwidth}
\begin{subfigure}[t]{\textwidth}
\vspace{3mm}
\centering
$
\scalebox{0.65}{
\begin{quantikz}[wire types = {b,b},column sep =4mm]
\lstick{} &  \gate{cs} \vqw{1}\gategroup[2,steps=3,style={inner
sep=2pt,fill=\gatecolor,draw=black},background]{$\proc$} &\gate[2]{l} &\gate[2,style={mystyle, inner sep=0pt}]{} &\\
\lstick{}  & \gate{\ST_1\ \ST_2} & & & 
\end{quantikz}
}
\longrightarrow  
\left\{
\scalebox{0.6}{
\begin{quantikz}[wire types = {b,b},column sep =4mm]
\lstick{} &  \gate{cs} \vqw{1}\gategroup[2,steps=3,style={inner
sep=2pt,fill=\gatecolor,draw=black},background]{$\proc$} &\gate[2]{l} &\gate[2,style={mystyle, inner sep=0pt}]{}&[2mm] \gate{cs} \vqw{1}\compgroup{2}{1}   &\\
\lstick{}  & \gate{\ST_1} & & & \gate{\text{S}_2} & 
\end{quantikz}
}
\text{\footnotesize { if $w_\PR^\proc(\ST_1)=1,$}}
\scalebox{0.6}{
\begin{quantikz}[wire types = {b,b},column sep =4mm]
\lstick{} &\gate{cs} \vqw{1}\compgroup{2}{1}   &[2mm] \gate{cs} \vqw{1}\gategroup[2,steps=3,style={inner
sep=2pt,fill=\gatecolor,draw=black},background]{$\proc$} &\gate[2]{l} &\gate[2,style={mystyle, inner sep=0pt}]{}  &\\
\lstick{}  &\gate{\ST_1}  & \gate{\text{S}_2}   & & &  
\end{quantikz}
}
\ \text{\footnotesize otherwise.}\right\}$
    \end{subfigure}
\begin{subfigure}[t]{\textwidth}
\vspace{5mm}
\centering
$
\begin{tabular}{c}
\scalebox{0.6}{
\begin{quantikz}[wire types = {b,b},column sep =4mm]
\lstick{} &  \gate{cs} \vqw{1}\gategroup[2,steps=3,style={inner
sep=2pt,fill=\gatecolor,draw=black},background]{$\proc$} &\gate[2]{l} &\gate[2,style={mystyle, inner sep=0pt}]{} &\\
\lstick{}  & \gate{\tb{if }b\tb{ then }\ST_\tb{true}\tb{ else }\ST_\tb{false}} & & & 
\end{quantikz}
}
\end{tabular}
\longrightarrow
\left\{
\scalebox{0.6}{
\begin{quantikz}[wire types = {b,b},column sep =4mm]
\lstick{} &  \gate{cs} \vqw{1}\gategroup[2,steps=3,style={inner
sep=2pt,fill=\gatecolor,draw=black},background]{$\proc$} &\gate[2]{l} &\gate[2,style={mystyle, inner sep=0pt}]{} &\\
\lstick{}  & \gate{\ST_b} & & & 
\end{quantikz}
}
\
\text{\footnotesize {if $w_\PR^\proc(\ST_b)=1,$}}
\scalebox{0.6}{
\begin{quantikz}[wire types = {b,b},column sep =3mm]
\lstick{}&  \gate[2]{l}\gategroup[2,steps=2,style={inner
sep=2pt,fill=\gatecolor,draw=black},background]{$\proc$}&\gate{cs}\vqw{1}\gategroup[2,steps=1,style ={mystyle, inner sep=0pt},background]{}& \\
\lstick{}&    & \gate{\text{S}_b} & 
\end{quantikz}
}
\ 
\text{\footnotesize otherwise.}\right\}
$
\end{subfigure}
\begin{subfigure}[t]{\textwidth}
\vspace{5mm}
\centering
$\begin{tabular}{c}
\scalebox{0.55}{
\begin{quantikz}[wire types = {b,q,b},row sep = 2mm,column sep=2mm]
\lstick{}&  \gate{cs} \vqw{1}\gategroup[3,steps=3,style={inner
sep=0pt,fill=\gatecolor,draw=black},background]{\proc} &\gate[3]{l}&\gate[3,style ={mystyle}]{}&  \\
\lstick{$n$}  &\gate[2, disable auto height]{\begin{matrix}\tb{qcase }\bar{\q}\el{n}\tb{ of }\\
\{0\to\ST_0,\,1\to\ST_1\}\end{matrix}} && & \\[-2mm]
  &  & & &
\end{quantikz}
}
\end{tabular}
\longrightarrow
\left\{
\scalebox{0.50}{
\begin{quantikz}[wire types = {b,q,b}, column sep=2mm]
& &   \gate[2]{cs[n:=b]} \vqw{1} \gategroup[3,steps=3,style={inner
sep=2pt,fill=\gatecolor,draw=black},background]{\proc} & \gate[3]{l}&[2mm] \gate[2]{cs[n:={1-b}]} \vqw{1} \gategroup[3,steps=1,style={mystyle, inner sep=-1pt},background]{} &  &\\
\lstick{$n$}& &  \vqw{1} & &\vqw{1} & &\\
&  & \gate{\text{S}_b}& &\gate{\text{S}_{1-b}} & & 
\end{quantikz}}
\scalebox{0.8}{
\begin{tabular}{c}
$ \text{if } \exists b \in \{0,1\}\ s.t.$\\
$ w_\PR^\proc(\ST_b)=1\text{ and}$\\
$w_\PR^\proc(\ST_{1-b})=0,$
\end{tabular}}
\scalebox{0.55}{
\begin{quantikz}[wire types = {b,q,b},column sep = 3mm]
\lstick{} &     \gate{cs} \vqw{1} \gategroup[3,steps=4,style={inner
sep=0pt,fill=\gatecolor,draw=black},background]{\proc} &  \gate{cs} \vqw{1} & \gate[3]{l}& \gate[3,style={inner
sep=0pt,mystyle}]{}  &  \\
\lstick{$n$} &  \octrl{1} & \ctrl{1} & & & \\
\lstick{}  & \gate{\text{S}_0}& \gate{\text{S}_1} & &  &
\end{quantikz}}\quad\text{\footnotesize{otherwise.}}\right\}
$
\end{subfigure}
\begin{subfigure}[t]{\textwidth}
\vspace{5mm}
\centering
$
\scalebox{0.8}{\begin{tabular}{c}
anchoring:
\end{tabular}}\begin{tabular}{c}
\scalebox{0.6}{
\begin{quantikz}[wire types = {b,b},column sep ={2.6mm}]
\lstick{}&   \gate{cs} \vqw{1}\gategroup[2,steps=3,style={inner
sep=0pt,fill=\gatecolor,draw=black},background]{\proc} &\gate[2]{l} &\gate[2,style ={mystyle}]{}& \\
   & \gate{\tb{call }\pn{proc'}(\sa);} & & & 
\end{quantikz}}
\end{tabular}
\longrightarrow
\scalebox{0.6}{
\begin{quantikz}[wire types = {q,b,b},column sep = {2.6mm}]
\inputD{a_{|\sa|}^{\pn{proc'}}} & \targ{} &  \ctrl{2} \gategroup[3,steps=3,style={inner
sep=0pt,fill=\gatecolor,draw=black},background]{\proc} & & & \targ{} &  \\
\lstick{}&  \gate{cs} \vqw{-1}  & & \gate[2]{l}&\gate[2,style ={mystyle}]{} &  \gate{cs} \vqw{-1} &  \\
\lstick{}&  & \gate{\text{S}^\pn{proc'}\{\sa/\bar{\q}\}}& & &  & 
\end{quantikz}}
$
\end{subfigure}
\begin{subfigure}[t]{\textwidth}
\vspace{5mm}
\centering
$\scalebox{0.8}{\begin{tabular}{c}
merging:
\end{tabular}}\begin{tabular}{c}
\scalebox{0.6}{
\begin{quantikz}[wire types = {q,b,b},column sep ={2.6mm}]
\lstick{$a_{|\sa|}^{\pn{proc'}}$} & \gategroup[3,steps=3,style={inner
sep=0pt,fill=\gatecolor,draw=black},background]{\proc}& \gate[3]{l} &\gate[3,style ={mystyle}]{}&\\
\lstick{}&   \gate{cs} \vqw{1} &&& \\
   & \gate{\tb{call }\pn{proc'}(\sa);} && & 
\end{quantikz}}
\end{tabular}
\longrightarrow
\scalebox{0.6}{
\begin{quantikz}[wire types = {q,b,b}, column sep = 2.1mm, row sep = 3.5mm]
\lstick{$a_{|\sa|}^{\pn{proc'}}$}& \targ{}  &  \gate[3]{l}\gategroup[3,steps=2,style={inner
sep=0pt,fill=\gatecolor,draw=black},background]{\proc}  & \gate[3,style={mystyle}]{} &  \targ{} & \\
\lstick{}&   \gate{cs} \vqw{-1}  & &  &  \gate{cs} \vqw{-1}  & \\
\lstick{}  & &    &  &  & 
\end{quantikz}}
$
\end{subfigure}
\begin{subfigure}[t]{\textwidth}
\vspace{5mm}
\centering
$\begin{tabular}{c}
\scalebox{0.6}{
\begin{quantikz}[wire types = {b}]
& \gate{[]}\gategroup[1,steps=2,style={inner
sep=0pt,fill=\gatecolor,draw=black},background]{\proc}  &\gate[3,style={mystyle}]{} &
\end{quantikz}
}
\end{tabular}
\longrightarrow
\scalebox{0.6}{
\begin{quantikz}[wire types = {b}] &\gate[3,style={mystyle}]{} &
\end{quantikz}}
$
\end{subfigure}
    \vspace{3mm}
\end{minipage}

\hrulefill

    \caption{Rewrite rules of \optimize{}.}
    \label{fig:optimize-steps}
\end{figure}

 The rewrite rules for the subroutine $\optimize$ on procedure $\proc$ are described in Figure~\ref{fig:optimize-steps}.  This subroutine takes two input parameters: 
 \begin{itemize} 
 \item a first list $l$ of controlled statements, the ones to be optimized,
 \item a second list of controlled statements, called \emph{contextual list} and denoted by $\begin{quantikz} \gate[style={mystyle}]{} \end{quantikz}$,  consisting in controlled statements that are orthogonal to those in $l$ and do not contain recursive procedure calls.
 \end{itemize}
 As described by the last rule of Figure~\ref{fig:compile}, each initial call to $\optimize$  is performed on the singleton list containing one controlled statement $(cs,\ST^{\proc}\{\sa/\bar{\q}\})$ and a contextual circuit equal to the identity circuit.

In Figure~\ref{fig:optimize-steps}, the rules are applied by considering the first controlled statement in the list $l$. We just specify the most interesting case below:
\begin{itemize}
\item for a controlled statement $(cs,\ST_1\ \ST_2)$, two distinct rules can be applied. In the first scenario, where $w_\PR^\proc(\ST_1)=1$, we have that $\ST_1$ contains a recursive procedure call and $\ST_2$ does not (this is a consequence of the $\wi_{\leq 1}$ condition in $\pbp$), or the converse. Depending on this, the rule select on which control statement $(cs,\ST_1)$ or $(cs,\ST_2)$ to perform the optimization and, append the compiled circuit of the other controlled statement to the left or to the right.
\item for a controlled statement with an \tb{if} statement, we precompute the boolean value $b$ (the value is computable), and according to whether the corresponding branch contains a recursive call or not, we either add it to $l$ or concatenate its compiled circuit to the contextual circuit.
\item for a controlled statement with a \tb{qcase} statement, one or two of the branches are added to the list $l$ depending on whether they are both recursive or not. 
\item for a controlled statement with a procedure call $\tb{call }\pn{proc'}(\sa);$ we perform either anchoring or merging depending on whether the ancilla $a_ {|\sa|}^{\pn{proc'}}$ exists or not. 
\item when the list of control statement is empty (i.e., $l=[]$), we compute the contextual list of controlled statements. This list is rearranged into lists of statements that are mutually recursive and compiled by a call to \optimize{}. Such a partition is important to perform merging of procedures with a recursion level lower than that of \proc, and we give further explanation on how it is performed in Appendix~\ref{app:compilation}. 
\end{itemize}

\subsection{Soundness of the algorithm}

In this section, we discuss the validity of the compilation algorithm. One first observation should be that, given a \pbp{} program, the compilation necessarily terminates. For instance, in \compile{} (Figure~\ref{fig:compile}), all rules besides the procedure call rewrite the controlled statement into either a circuit or into instances of \compile{} of smaller statements. In the case a procedure call, the rewriting of the procedure body produces a finite number of calls to procedures of lower recursive level.

In \optimize{}, a recursive procedure will result in a finite number of calls to mutually recursive procedures -- this is ensured by the well-foundedness condition $\wf{}$, that requires that recursive procedure calls reduce the size of the input, therefore procedure calls either reduce the level of recursion or the number of available qubits.

The soundness of the compilation algorithm is ensured by an orthogonality invariant in \optimize{}. Let $cs,\,cs'$ be two control structures. We say that $cs$ and $cs'$ are \emph{orthogonal} if there exists $ i \in\text{dom}(cs)\cap \text{dom}(cs')$ such that $cs(i)  =1 -cs'(i)$. Two controlled statements are orthogonal if their control structures are orthogonal.

\begin{lemma}
\label{lem:inv}
During the compilation of a \pbp{} program $\PR$, for each optimization of a (recursive) procedure $\proc$, all controlled statements in the union of list $l$ and the contextual list $\contL$ are pairwise orthogonal.
\end{lemma}

This invariant ensures the validity of \optimize{}, and the soundness of the compilation algorithm. It is also a consequence of the $\wi_{\leq 1}$ restriction in $\pbp$, as by the definition of width, at most one recursive call may appear per branch of a recursive procedure. This ensures that two recursive calls on a given procedure always occur in orthogonal branches and can be simply combined in the same ancilla. Given a circuit $C$, we define its semantics $\llbracket C\rrbracket$ naturally as the composition of the semantics of each gate.

\begin{restatable}[Soundness of compilation]{theorem}{thmsoundness}\label{thm:soundcompile}
Given a \pbp{} program $\PR$ and a quantum state $\ket{\psi}\in\mathcal{H}_{2^n}$ we have that $\llbracket \compile{}(\PR,n)\rrbracket(\kpsi)=\llbracket \PR\rrbracket(\ket{\psi})$.
\end{restatable}

\subsection{Illustrating Example}

\begin{wrapfigure}{r}{0.44\textwidth}
\centering
\vspace{-7mm}
\begin{tabular}{rl}
\ttt{1}&$\tb{decl }\pn{rec}(\bar{\q})\{$\\
\ttt{2}&$\quad \tb{if }|\bar{\q}|>2\tb{ then }$\\
\ttt{3}&$\qquad \tb{qcase }\bar{\q}\el{1}\tb{ of}\{$\\
\ttt{4}& $\qquad \quad 0 \to \tb{call }\pn{rec}(\bar{\q}\ominus\el{1});$\\
\ttt{5}& $\qquad \quad 1\to \tb{qcase }\bar{\q}\el{2}\tb{ of}\{$\\
\ttt{6}& $\qquad  \qquad\qquad 0\to \tb{skip};$\\
\ttt{7}& $\qquad \qquad\qquad 1 \to \tb{call }\pn{rec}(\bar{\q}\ominus\el{1,2});$\\
\ttt{8}&$\qquad\qquad\qquad \}$\\
\ttt{9}&$\qquad\quad\}$\\
\ttt{10}&$\quad \tb{else }\bar{\q}\el{1}\asg \U;\,\}$\\
\ttt{11}&$::\ \tb{call }\pn{rec}(\bar{\q});$\\
\end{tabular}
\caption{Program \texttt{REC}.}
\label{fig:example-program}
\vspace{-9mm}
\end{wrapfigure}

We illustrate the compilation process with the program  $\texttt{REC}$, of Figure~\ref{fig:example-program}. 
Notice that $\texttt{REC}$ is in $\wf \cap \wi_{\leq 1}$ but does not belong to $\pbp$, as the procedure $\pn{\texttt{rec}}$ performs two recursive calls on different sorted sets $\bar{\q}\ominus\el{1}$ and $\bar{\q}\ominus\el{1,2}$. This example thus illustrates that the compilation algorithm is not restricted to $\pbp$. 

Given that \pn{rec} is a recursive procedure, its compilation is performed within \optimize{}.
We denote by the empty box $\begin{quantikz} \gate{} \end{quantikz}$ the procedure statement $\ST^\pn{rec}$ and by the dotted box $\begin{quantikz} \gate[style={mystyle2}]{} \end{quantikz}$ the statement $\tb{call }\pn{rec}(\sa);$. Applying rewrite rules to the statement of \pn{rec} we obtain the transitions:
%

\begin{center}
\includegraphics[scale=1]{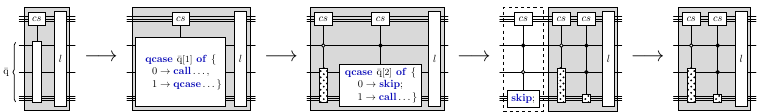}
\end{center}

The first step is obtained by applying the rule of $\tb{if}$ statements, where we assume that $|\bar{\q}|>2$. Afterwards, we apply twice the rule for the $\tb{qcase}$ statement, adding the two qubits $\bar{\q}\el{1}$ and $\bar{\q}\el{2}$ to the control structure. Finally, the compilation of the $\tb{skip}$ statement corresponds to the empty circuit. We denote by $\xrightarrow[]{\texttt{rec}}$ the application of this sequence of rewrite rules, which is used in Figure~\ref{fig:compilation-example}.

\addtolength{\tabcolsep}{0pt}   
\begin{figure}[h!]
\hrulefill

\centering
\includegraphics[scale=1]{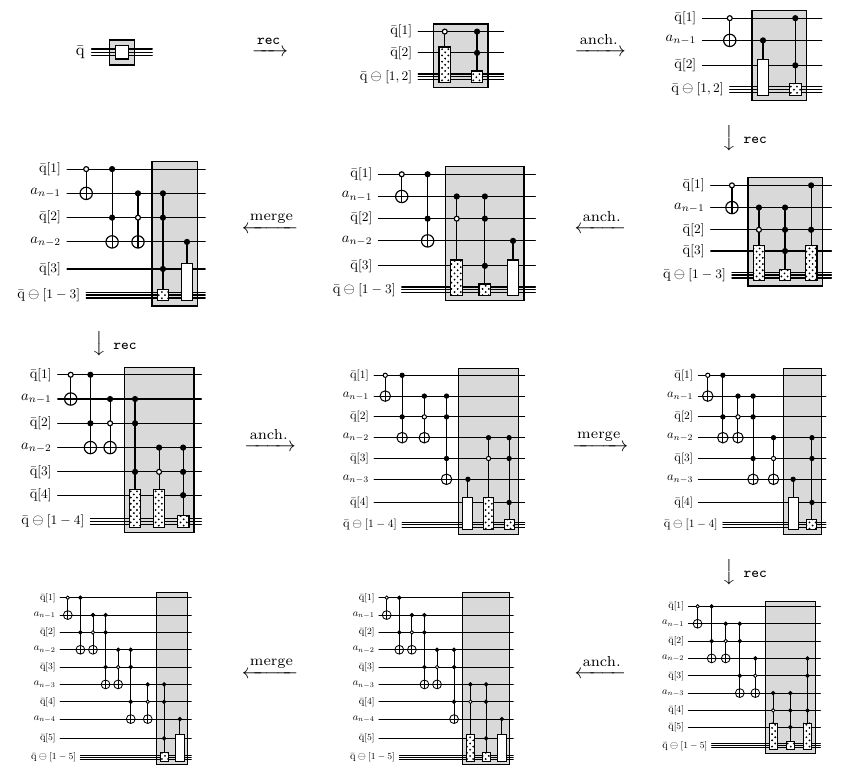}

\hrulefill
\caption{Example of anchoring and merging for the program \texttt{REC} in Figure~\ref{fig:example-program}.}
\label{fig:compilation-example}
\end{figure}

One important thing to notice in Figure~\ref{fig:compilation-example} is that the compilation strategy does not avoid branch sequentialization \emph{locally} but rather \emph{asymptotically} in the construction of the circuit. In other words, the \tb{qcase} statement in rule \texttt{rec} does generate two instances of $\ST^\pn{rec}$ in the circuit in sequence, one for each branch. However, the merging of calls to $\pn{rec}$ on inputs of the same size ensures that only one instance per input size needs to be compiled, and therefore this strategy achieves linear complexity in the number of gates.

\section{Main Results}

\subsection{No Branch Sequentialization}
First, we show that the problem of branch sequentialization is solved in $\pbp$. For that purpose, we show that the size of quantum circuit obtained through compiling a $\pbp$ program is bounded asymptotically by the time complexity of the program.
As the time complexity is the maximum of the complexity of the two branches of a quantum case statement, the branch sequentialization problem is avoided on $\pbp$ programs, as the size of the circuit is asymptotically equal to its time semantics. Given a circuit $C$, let $\# C$ denote its size, i.e., number of gates and wires.

\begin{restatable}[No branch sequentialization]{theorem}{thmnewcompilation}
\label{thm:new-compilation}
Given a program $\PR\in\pbp{}$ and input $n\in\mathbb{N}$, we have that $\# \textnormal{\compile{}}(\PR,n)=O(\level_\PR(n))$ holds.
\end{restatable}

One may notice that, in the rules of $\compile$ (Figure~\ref{fig:compile}), the  compilation of a \tbt{qcase} statement generates two controlled statements in sequence. Similarly, the rule of $\optimize$ (Figure~\ref{fig:optimize-steps}) for $\tb{qcase}$ statements appends two controlled statements to the list $l$ in the case where they are both recursive. This does not pose a problem as, anchoring and merging will ensure that the asymptotic complexity of this statement will be given by the maximum complexity of the branches.

\begin{example}
By Theorem~\ref{thm:new-compilation} and Example~\ref{ex:pairs}, we have that
$\#\compile{}(\textnormal{\texttt{PAIRS}},n)=O(n)$.
\end{example}

\subsection{Polynomial-Time Soundness and Completeness}

Having shown in Theorem~\ref{thm:new-compilation} that \pbp{} programs can avoid branch sequentialization, we now turn to the question of whether they constitute an interesting fragment of quantum programs, given that it is restricted by the  $\wf$, $\wi_{\leq 1}$ and $\unif{}$ conditions. We show that the set of \pbp{} programs is sound and complete for quantum polynomial time via an implicit characterization of \fbqp{}, i.e. the set of classical functions that can be approximated with bounded error by a polytime quantum Turing machine~\cite{yamakami2022,HPS23}.

Since the considered programming language does not allow for qubit creation, in order to define functions where the output is larger than the input, we make use of polytime padding. A \emph{polytime padding function} is a function $f:\{0,1\}^\ast\to \{0,1\}^\ast$ computable by a (classical) polytime Turing machine such that $f(x)=xy$, for some $y$ depending only on the length of $x$ (and not the value of $x$). Given a set $\mathcal{F}$ of programs, $\llbracket \mathcal{F} \rrbracket$ denotes the set of functions given by $\llbracket \PR \rrbracket \circ f$, where $\PR\in \mathcal{F}$, $f$ is a polytime padding function, and $\circ$ is the standard function composition. For any $p\in(\frac{1}{2},1]$, we denote by $\llbracket \mathcal{F} \rrbracket_{\geq p}$ the set of functions in $\llbracket \mathcal{F} \rrbracket$ that approximate a classical function with probability at least $p$.


\begin{restatable}[Soundness and Completeness]{theorem}{thmsoundnesscompleteness}
\label{thm:soundness-and-completeness}
$\llbracket \pbp{} \rrbracket_{\geq\frac{2}{3}} = \fbqp{}$.
\end{restatable}

\subsection{Extension to Programs in $\wf \cap \wi_{\leq 1}$}
\label{ss:compwfwi}

The \unif{} restriction  of $\pbp$ can be thought of as ensuring that qubits are consumed in a uniform way. As a consequence, any two procedure calls on inputs of size $n$ will correspond to precisely the same $n$ qubits. This guarantees that two compatible procedure calls can be merged simply by combining control structures in the same ancilla.
Without this constraint (i.e., on $\wf \cap \wi_{\leq 1}$), merging can be done by control-swapping registers into a single procedure call. For $k$ instances of a procedure on input size $n$, this requires $k-1$ controlled swaps, each requiring $O(n)$ operations in the worst case. An example of merging two unitary gates $U$ with controlled-swaps is given in Figure~\ref{fig:merging-with-control-swap}. If we allow for parallelization of gates in the circuit, this can be done in logarithmic depth, as shown in Lemma~\ref{lem:permutation}. An example of a controlled permutation is shown in Figure~\ref{fig:example-of-controlled-permutation}, where vertical dashed lines separate time slices where gates can be applied simultaneously.

\begin{figure}
\centering
\scalebox{0.65}{
$
\left\llbracket
\begin{quantikz}[row sep = {6mm,between origins}, column sep = 4mm]
& \ctrl{1} & \octrl{1} &\\
& \ctrl{1} & \gate[3]{U} &\\
 & \gate[3]{U} & &\\
 & & &\\
 & &\ctrl{-1} &\\
 \lstick{$\ket{0}$} & & &\\
 \lstick{$\ket{0}$} & & &
\end{quantikz}
\right\rrbracket
=
\left\llbracket
\begin{quantikz}[row sep = {6mm,between origins}, column sep = 4mm]
& \ctrl{1} &\octrl{4} & &  &&& &\octrl{4} & \ctrl{1} &\\
& \ctrl{4} & & &\gate[4]{}\permute{4,1,2,3} & \gate[3]{U}&\gate[4]{}\permute{2,3,4,1} & & & \ctrl{4} &\\
 &  & && && & & & &\\
 & & && && & & & &\\
 & &\ctrl{2}& & & && & \ctrl{2} & &\\
 \lstick{$\ket{0}$}& \targ{} &  &\ctrl{1}& \ctrl{-1} &  &\ctrl{-1}&\ctrl{1} & & \targ{} &\\
 \lstick{$\ket{0}$}&   &  \targ{} & \targ{} & & \ctrl{-3} && \targ{} & \targ{} & &
\end{quantikz}
\right\rrbracket
$}
\caption{Merging orthogonal controlled statements.}
\label{fig:merging-with-control-swap}
\end{figure}

\begin{figure}[t]
\centering
\includegraphics[scale=0.8]{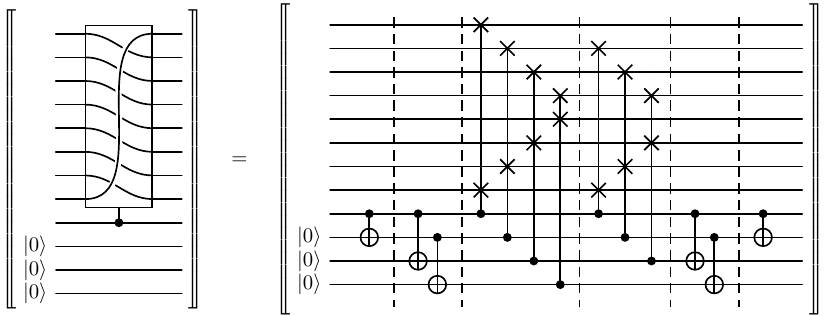}
\caption{Implementation of a controlled permutation in log-depth.}
\label{fig:example-of-controlled-permutation}
\end{figure}

\begin{restatable}{lemma}{permutation}
\label{lem:permutation}Any controlled permutation on $n$ qubits can be performed with a quantum circuit of size $O(n)$ and depth $O(\log n)$.
\end{restatable}


An extension of \compile{} with controlled-swapping as in Figure~\ref{fig:merging-with-control-swap} used for merging procedures ensures the following bound for $\wf \cap \wi_{\leq 1}$ programs.

\begin{theorem}\label{thm:non-basic-compilation}
For $\PR\in\wf \cap \wi_{\leq 1}$, $\textnormal{\compile{}}(\PR,n)$ results in a circuit with ${O}(\log(n)\cdot\level_\PR(n))$ depth and $O(n\cdot \level_\PR(n))$ size.
\end{theorem}

%
%
%

\subsection{Comparison with Existing Work}
The $\compile$ algorithm strictly improves upon the size bounds of other compilation algorithms~\cite{HPS23}, as illustrated by Table~\ref{table:problems}. In some cases, such as Examples~\ref{ex:chained-substring}~and~\ref{ex:full-adder} given in Appendix~\ref{app:examples-table}, we find families of programs where the gap in complexity (i.e. the degree of the polynomial) can be made arbitrarily large.

In $k$-Chained Substrings (Example~\ref{ex:chained-substring}), we consider the problem of detecting whether an input string contains substrings $y_1,\dots, y_k$ appearing in that order. For a certain choice of substrings $y_i$, we define a \pbp{} program of linear time complexity for which the compilation strategy in~\cite{HPS23} results in circuits of size $\Theta(n^{3k})$. The problem Sum($r$), given in Example~\ref{ex:sumr}, consists of checking if an input $x=x_1\dots x_n$ satisfies $\sum_{i=1}^n x_i=r$.



\begin{table}[h]
\renewcommand{\arraystretch}{1.5}
\noindent
\begin{tabularx}{\textwidth} { 
  | >{\centering\arraybackslash\hsize=0\hsize}X 
   | >{\centering\arraybackslash\hsize=1.2\hsize}X 
  | >{\centering\arraybackslash\hsize=0.9\hsize}X 
  | >{\centering\arraybackslash\hsize=0.9\hsize}X | }
 \cline{3-4}
 \multicolumn{2}{l|}{}&  \multicolumn{2}{c|}{\tbt{Circuit complexity}}\\
 \hline
  \multicolumn{1}{|l|}{\tbt{Problem}}&\tbt{Example}  &\cite{HPS23} & \compile{} \\
    \hline
 \multicolumn{1}{|l|}{Full Adder } &Example~\ref{ex:full-adder}& $\Theta(n)$ & $\Theta(n)$ \\
\hline
 \multicolumn{1}{|l|}{Quantum Fourier Transform } &Example~\ref{ex:qft}   & $\Theta(n^2)$ & $\Theta(n^2)$ \\
 \hline
  \multicolumn{1}{|l|}{$k$-Chained Substrings}& Example~\ref{ex:chained-substring} & $\Theta(n^{3k})$  & $\Theta(n)$ \\
\hline
 \multicolumn{1}{|l|}{Sum($r$), $r\geq 1$} & Example~\ref{ex:sumr} & $\Theta(n^r)$ & $\Theta(n)$ \\
\hline
 \end{tabularx}
\caption{Circuit size complexity bounds.}
\label{table:problems}
\end{table}

%
%


\section{Conclusions and Future Work}

The quantum control statement, while being a central component of programming languages with quantum control flow, usually incurs a remarkable slowdown in the program complexity in any automatic implementation, which poses a problem to the quantum programmer. In this paper, we have demonstrated that such a slowdown can be avoided by restricting the program structure. This was achieved using techniques from quantum implicit computational complexity, which allow not only to selectively reason about efficient programs (polytime soundness) but also to ensure that the techniques are sufficient in principle for any program (polytime completeness).

$\pbp$ is then the first quantum programming language with a compilation strategy that ensures that the quantum control statement can be compiled onto a circuit whose size (i.e., number of gates and wires) is bounded asymptotically by the maximum of the time complexity of its branches. While this language is extensionally complete for {\fbqp} and expressive enough to write quantum algorithms such as QFT or Full Adder in a natural way, it would be interesting to investigate in what ways its expressive power can be improved, and how different languages can also be shown to avoid branch sequentialization.

\section{Acknowledgments}

This work is supported by the the Plan France 2030 through the PEPR integrated project
EPiQ ANR-22- PETQ-0007 and the HQI platform ANR-22-PNCQ-0002; and by the European
projects NEASQC and HPCQS.

\bibliography{references}

\newpage

\appendix

\section{Semantics of expressions}
\begin{figure}[!h]
\centering
\fbox{
\begin{minipage}{0.95\textwidth}
\vspace*{5mm}
\begin{center}
\begin{prooftree}
\hypo{(\ea, l) \Downarrow_{\sem{\tau_1}} m}
\hypo{(\da, l) \Downarrow_{\sem{\tau_2}} n}
\infer2[(Op)]{(\ea\ \op\ \da, l) \Downarrow_{\sem{\op}(\sem{\tau_1},\sem{\tau_2})} \sem{\op}(m,n)}
\end{prooftree}
\qquad
  \begin{prooftree}
    \hypo{(\ia,l) \Downarrow_{\Z} n}
    \infer1[(Unit)]{(\U^f(\ia), l) \Downarrow_{\mathbb{C}^{2\times2}} \sem{\U^f}(n)}
  \end{prooftree}
\end{center}

\begin{center}
\begin{prooftree}
\hypo{}
\infer1[(Cst)]{(n,l)\Downarrow_{\Z} n}
\end{prooftree}
\qquad
\begin{prooftree}
\hypo{(\sa, l)\Downarrow_{\mathcal{L}(\mathbb{N})} [x_1, \ldots, x_m]}
\hypo{(\ia, l) \Downarrow_{\Z} k \in [1,m]}
\infer2[(Rm$_\in$)]{(\sa\ominus\el{\ia}, l) \Downarrow_{\mathcal{L}(\mathbb{N})} [x_1, \ldots, x_{k-1}, x_{k+1}, \ldots, x_m]}
\end{prooftree}
\end{center}

\begin{center}
\begin{prooftree}
\hypo{(\sa,l)\Downarrow_{\mathcal{L}(\mathbb{N})} [x_1,\ldots,x_n]}
\infer1[(Size)]{(\size{\sa},l)\Downarrow_{\Z} n}
\end{prooftree}
\qquad
\begin{prooftree}
\hypo{(\sa, l)\Downarrow_{\mathcal{L}(\mathbb{N})} [x_1, \ldots, x_m]}
\hypo{(\ia, l) \Downarrow_{\Z} k \notin [1,m]}
\infer2[(Rm$_{\notin}$)]{(\sa\ominus\el{\ia}, l) \Downarrow_{\mathcal{L}(\mathbb{N})} [\,]}
\end{prooftree}
\end{center}

\begin{center}
\begin{prooftree}
\infer0[(Var)]{(\bar{\q}, l) \Downarrow_{\mathcal{L}(\mathbb{N})} l}
\end{prooftree}
\qquad
\begin{prooftree}
\hypo{(\sa, l)\Downarrow_{\mathcal{L}(\mathbb{N})} [x_1, \ldots, x_m]}
\hypo{(\ia, l) \Downarrow_{\Z} k \in [1,m]}
\infer2[(Qu$_\in$)]{(\sa\el{\ia}, l) \Downarrow_{\mathbb{N}} x_k}
\end{prooftree}
\end{center}

\begin{center}
\begin{prooftree}
\hypo{(\sa, l)\Downarrow_{\mathcal{L}(\mathbb{N})} [x_1, \ldots, x_m]}
\hypo{(\ia, l) \Downarrow_{\Z} k \notin [1,m]}
\infer2[(Qu$_{\notin}$)]{(\sa\el{\ia}, l) \Downarrow_{\mathbb{N}} 0}
\end{prooftree}
\end{center}
\vspace*{3mm}

\end{minipage}}

\caption{Semantics of expressions.}

\label{table:semnatbool}
\end{figure}

\section{Compilation}
\label{app:compilation}

In this appendix, we describe the \compile{} and \optimize{} routines in more detail.
For simplicity, we provide here the algorithm for programs in $\pbp$.
The version for $\wf\cap\wi_{\geq 1}$ uses additional control swaps in the {\optimize} subroutine and is described in Section~\ref{ss:compwfwi}.
The full definition of \compile{} (Algorithm~\ref{alg:comprec}) takes four inputs: a program $\PR \in \wf \cap \wi_{\leq 1}$, a list of qubit pointers $l$ and a control structure $cs$. We make use of the following syntactic sugar to denote the initial call to \compile{}:
 $$\compile{}(\PR,n)\triangleq \compile{}(\PR,[1,\dots,n],\cdot),$$ 
\noindent where $\PR$ is the program to be compiled, $[1,\ldots,n]$ is list of qubit pointers (initially all qubits), $\cdot$ is an empty control structure, and $\{\}$ an empty dictionary.

\captionsetup{justification=raggedright}

\begin{algorithm}[h]
\begin{algorithmic}[1]
\If{$\ST= \skp$}
\State{$C\leftarrow \mathds{1}$}\Comment{Identity circuit}
\\
\ElsIf{$\ST= \q\asg\U^f(\ja);$\textbf{ and }$(\q,l)\Downarrow_{\mathbb{N}}n$\textbf{ and }$(\U^f(\ja),l)\Downarrow_{\mathbb{C}^{2\times 2}}M$}
\State{$C\leftarrow M(cs,[n])$}\Comment{Controlled gate}
\\
\ElsIf{$\ST= \ST_1\,\ST_2$}
\State{$C\leftarrow \compile(\D::\ST_1,l,cs)\circ \compile(\D::\ST_2,l,cs)$}\Comment{Composition}
\\
\ElsIf{$\ST=\tb{if }\bb\tb{ then }\ST_{\tb{true}}\tb{ else }\ST_{\tb{false}} \textbf{ and }(\bb,l)\Downarrow_{\B}b$}
\State{$C\leftarrow \compile(\D::\ST_b,l,cs)$}\Comment{Conditional}
\\
\ElsIf{$\ST=\qcase{\sa\el{\ia}}{\ST_0}{\ST_1}$\textbf{ and }$(\sa\el{\ia},l)\Downarrow_{\mathbb{N}}n$}\Comment{Quantum case}
\State{$C \leftarrow \compile(\D::\ST_0,l,cs[n:=0])\circ \,\compile(\D::\ST_1,l,cs[n:=1])$}
\\
\ElsIf{$\ST= \call \proc(\sa);$\textbf{ and }$(\sa,l)\Downarrow_{\mathcal{L}(\mathbb{N})}[\,]$}
\State $ C\leftarrow\mathds{1}$\Comment{Nil call}
\\
\ElsIf{$\ST= \call \proc(\sa);$\textbf{ and }$(\sa,l)\Downarrow_{\mathcal{L}(\mathbb{N})}l'\not=[\,]$}
\If{$w_\PR^\proc(\ST^\proc)=0$} \Comment{Non-recursive procedure call}
\State $C\leftarrow \compile(\D::\ST^\proc,l',cs)$
\Else
\State $ C\leftarrow\optimize(\D,[(cs,\ST^{\proc},l)], \proc)$
\EndIf
\EndIf\\
\Return{C}
\end{algorithmic}
\caption{(\compile)\\
\textbf{Input:} $(\D::\ST,l,cs)\in \text{Programs}\times \LN \times (\mathbb{N} \to \{0,1\})$}
\label{alg:comprec}
\end{algorithm}

The aim of $\compile{}$ is to generate the quantum circuit corresponding to $\PR$ on $n$ qubits inductively on the program statement of $\PR$. When the analyzed statement is a recursive procedure call, $\compile$ calls the $\optimize$ subroutine (Algorithm~\ref{alg:optimize}) to perform an optimization of the generated quantum circuit via anchoring and merging. $\optimize$ has the same inputs as $\compile$ with the addition of a list of \emph{controlled statements} $l_\CST$ and the name $\proc$ of the procedure under analysis.

As described in Subsection~\ref{ss:algorithm}, {\optimize} manages a contextual list $\contL$ of controlled statements that do not contain recursive procedure calls.
At the end of \optimize{}, we rearrange the contextual list 
in the following way. Let $\contL{}=[(cs_1,\ST_1,l_1),\dots,(cs_k,\ST_k,l_k)]$ be the state of the contextual list at the end of \optimize{}. We may rewrite each controlled statement as a list of its atomic elements. This transformation, denoted $\text{seq}$, can be described inductively:

\begin{align*}
\text{seq}(cs,\tb{skip};,l) &\triangleq\  [],\\
\text{seq}(cs,\q\asg \text{U}^f(\ja);,l) &\triangleq\  [(cs,\q\asg \text{U}^f(\ja);,l)],\\
\text{seq}(cs,\ST_1\ \ST_2,l) &\triangleq\  \text{seq}(cs,\ST_1,l) @ \text{seq}(cs,\ST_1,l),\\
\text{seq}(cs,\tb{if }\bb\tb{ then }\ST_\tb{true}\tb{ else }\ST_\tb{false},l) &\triangleq\  \text{seq}(cs,\ST_b,l),\text{ if }(\bb,l)\Downarrow_\mathbb{B} b,\\
\text{seq}(cs,\tb{qcase }\q\tb{ of }\{0\to\ST_0,\,1\to\ST_1\},l) &\triangleq\  \text{seq}(cs[n:=0],\ST_0,l)@\text{seq}(cs[n:=1],\ST_1,l)\text{ if }(\q,l)\Downarrow_\mathbb{N} n,\\
\text{seq}(cs,\tb{call }\proc(\sa);,l) &\triangleq [(cs,\tb{call }\pn{proc}(\sa);,l)].
\end{align*}

This separation of statements allows for a partitioning according to the type of procedure call appearing in the statement. Given a list of controlled statements $\mathcal{L}$, we denote by procedure\_split($\mathcal{L}$) the list $[\mathcal{L}_0,\mathcal{L}_1,\dots,\mathcal{L}_m]$ where, for \pn{proc$_1$}, \dots, \pn{proc$_m$} are procedures that are not mutually recursive.
\begin{align*}
\mathcal{L}_0 &\triangleq \{ (cs,\ST,l)\in \mathcal{L}: \not\exists \proc \text{ such that }w_\PR^\proc(\ST^\proc)=1\text{ and }w_\PR^\proc(\ST)=1\}.\\
\mathcal{L}_\pn{proc$_i$} &\triangleq \{ (cs,\ST,l)\in \mathcal{L}: w_\PR^{\pn{proc$_i$}}(\ST)=1\},\quad i=1,\dots,m.
\end{align*}

Given our choice of procedures and the controled statements obtained from seq, this partition is well-defined. Performing these two partitions (first in terms of sequential order of statements, and then according to procedure calls), we are able to rewrite the list $\contL$ in the following way:
\begin{center}
\includegraphics[scale=0.6]{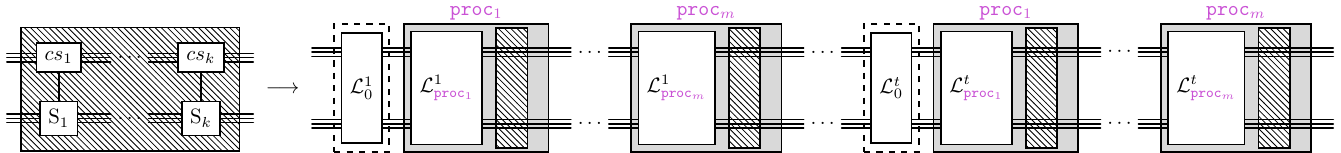}
\end{center}
The different instances of \optimize{} contain calls to procedures that are mutually recursive, which will allow for further anchoring and merging.

\begin{algorithm}
\caption{\textbf{(\optimize)} Build circuit for recursive procedure $\proc$ \\
\textbf{Inputs:} $(\D,l_{\CST},\proc)\in \textnormal{Decl}\times \mathcal{L}(\CST\times \LN) \times \textnormal{Procedures}$}\label{alg:optimize}

\begin{algorithmic}[1]
\State $C_{\text{L}}\leftarrow\mathds{1}$; $C_{\text{R}}\leftarrow\mathds{1}$;$C_{\text{M}}\leftarrow\mathds{1}$; $\contL\leftarrow [\,]$; Anc$\leftarrow \{\,\}$; $\PR\leftarrow \D::\tb{skip};$

\While{$l_{\CST}\neq [\,]$}
\State $(cs,\ST,l) \leftarrow hd(l_{\CST});\ l_{\CST} \leftarrow tl(l_{\CST})$\\

\If{$\ST=\ST_1\,\ST_2$}
\If{$w_\PR^\proc(\ST_1)=1$}
\State $l_{\CST}\leftarrow l_{\CST}\MVAt[(cs,\ST_1,l)];\  C_\text{R}\leftarrow \compile(\D::\ST_1,l,cs)\circ C_\text{R}$
\Else
\State $l_{\CST}\leftarrow l_{\CST}\MVAt[(cs,\ST_2,l)];\  C_\text{L}\leftarrow C_\text{L} \circ\compile(\D::\ST_1,l,cs)$
\EndIf
\EndIf\\

\If{$\ST=\tb{if }\bb\tb{ then }\ST_{\tb{true}}\tb{ else }\ST_{\tb{false}}\textbf{ and } (\bb,l) \Downarrow_\B b$}
\If{$w_\PR^\proc(\ST_b)=1$}
\State $l_{\CST}\leftarrow l_{\CST}\MVAt[(cs,\ST_b,l)]$
\Else
\State  $\contL\leftarrow \contL \MVAt[(cs,\ST_b,l)]$
\EndIf
\EndIf\\

\If{$\ST=\qcase{\q}{\ST_0}{\ST_1}$\textbf{ and }$(\q,l)\Downarrow_{\mathbb{N}}n$}
\If{$w_\PR^\proc(\ST_0)=1 \textbf{ and } w_\PR^\proc(\ST_1)=1$}
\State $ l_{\CST}\leftarrow l_{\CST}\MVAt[(cs[n:=0],\ST_0,l),(cs[n:=1],\ST_1,l)]$
\ElsIf{$w_\PR^\proc(\ST_1)=0$}
\State $ l_{\CST}\leftarrow l_{\CST}\MVAt[(cs[n:=0],\ST_0,l)];\ \contL\leftarrow\contL\MVAt[(cs[n:=1],\ST_1,l)]$
\ElsIf{$w_\PR^\proc(\ST_0)=0$}
\State $l_{\CST}\leftarrow l_{\CST}\MVAt[(cs[n:=1],\ST_1,l)];\ \contL\leftarrow\contL\MVAt[(cs[n:=0],\ST_0,l)]$
\EndIf
\EndIf\\

\If{$\ST=\call \pn{proc'}(\sa);$\tbt{ and }$(\sa,l)\Downarrow_{\mathcal{L}(\mathbb{N})}l'\not=[\,]$
}
\If{$(\pn{proc'},|l'|)\in\No$} 
\State{\tbt{Let }$a = \No[\pn{proc'},|l'|]$\tbt{ in}} \textcolor{gray}{/* compatible procedure exists: merging */}
\State $C_\text{L}\leftarrow C_\text{L} \circ NOT(cs,a);$
\State $C_\text{R}\leftarrow NOT(cs,a) \circ C_\text{R};$ 
\Else 
\State $a \leftarrow \tbt{new }ancilla()$ \textcolor{gray}{/* no compatible procedure: anchoring */}
\State $\No[\pn{proc'},|l'|]\leftarrow a;$
\State{$C_\text{L}\leftarrow C_\text{L}\circ NOT(cs,a);\ C_\text{R}\leftarrow NOT(cs,a)\circ C_\text{R};$}
\State{$l_{\CST}\leftarrow l_{\CST}\MVAt[(\cdot[a=1],\ST^{\pn{proc'}},l')]$}

\EndIf
\EndIf
\EndWhile
\State $T=\max_{(cs,\ST,l)\in\contL}(|\text{sec}(cs,\ST,l)|)$
\For{$1\leq t\leq T$}
\State{$\mathcal{L}\leftarrow \bigcup_{(cs,\ST,l)\in\contL} \text{sec}(cs,\ST,l)[t]$}
\State{$\mathcal{L}_0,\dots \mathcal{L}_m = \text{procedure\_split}(\mathcal{L})$\textcolor{gray}{/* $m =$ number of recursive procedure families */}}
\State{$C_\text{M} \leftarrow C_\text{M} \circ \big(\circ_{(cs,\ST,l)\in\mathcal{L}_0}  \compile(\D::\ST,l,cs) \big)\circ \big(\circ_{i=1}^m \optimize(\D,\mathcal{L}_i,\pn{proc$_i$})\big)$}
\EndFor\\
\Return{$C_\text{L} \circ C_\text{M}\circ C_\text{R}$}
\end{algorithmic}
\end{algorithm}

\section{Proofs}
\label{app:proofs}

\thmsoundness*
\begin{proof}
By induction on the structure of the \pbp{} program $\PR=\D::\ST$. One can check by inspection of each case that the \compile{} rules for non-recursive statements corresponds to the straightforward circuit semantics of quantum programs.

Likewise, the rewriting rules of \optimize{}, given in Figure~\ref{fig:optimize-steps}, can be easily checked using the orthogonality invariant of Lemma~\ref{lem:inv}. For instance, consider the case of the sequential statement $\ST = \ST_1\ \ST_2$. In the case where $w^\proc_\PR(\ST_1)=1$, we can derive the rule for the statement with the following steps:

\begin{center}
\includegraphics[scale=1]{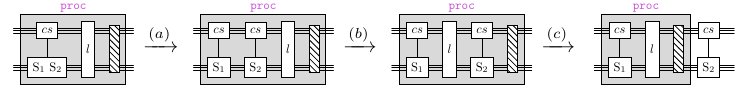}
\end{center}

\noindent
where $(a)$ corresponds to the definition of the sequential statement, in $(b)$ we make use of the fact that $(cs,\ST_2)$ is orthogonal to all controlled statements in $l$ and therefore can be commuted in the circuit, and in step $(c)$ we likewise make use of orthogonality, in this case with the controlled statements in the contextual list. Other case can be inspected to follow a short sequence of steps, such as described for the sequential statement.

For instance, in the case of anchoring, we consider the following composition of rules

\begin{center}
\includegraphics[scale=1]{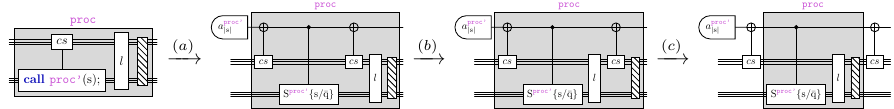}
\end{center}

\noindent
where, likewise, $(a)$ corresponds to the typical circuit semantics of a procedure call (with an added anchoring ancilla), and $(b)$ and $(c)$ make use of the orthogonality of $cs$ with the right side of the circuit. The validity of merging and also be checked:
\begin{center}
\includegraphics[scale=1]{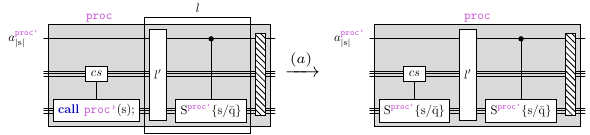}

\includegraphics[scale=1]{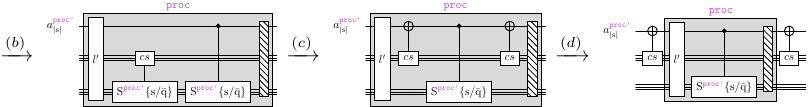}
\end{center}

Step $(a)$ can be seen as the definition of the procedure call. Since we are in the merging scenario, a similar procedure call has already been performed and is anchored to its corresponding ancilla. Without loss of generality (since all controlled statements in $l$ are orthogonal and therefore commute) we assume that this call appears at the end of $l$ (we reuse the name $l$ in the circuit rewriting here as an abuse of notation). Rule $(b)$ simply indicates that since $cs$ is orthogonal to other control structures in $l$, we may move the controlled statements so that they are adjacent, and where we apply step $(c)$ to perform merging. Since $cs$ is orthogonal to $\cdot[a_{|\sa|}^{\pn{proc'}}=1]$ the control is added to the ancilla as expected. Finally, orthogonality of $cs$ with other control structures means that we may move the two controlled-$NOT$s to the edges of the circuit.

Finally, we consider the validity of the rule for the contextual circuit. Consider a list of mutually-orthogonal controlled statements, where the sequential form of $(cs_i,\ST_i)$ is given by the lists of controlled statements $l^{(i)}_1,\dots,l^{(i)}_t$, then we have that:

\begin{center}
\includegraphics[scale=1]{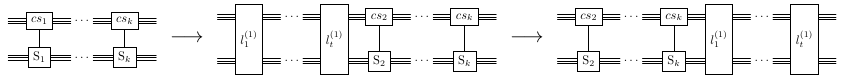}
\end{center}

\noindent
The final step comes from the fact that all $l^{(1)}_j$ are orthogonal to $cs_2$, \dots, $cs_k$ since they include $cs_1$. Using the sequential form of $(cs_2,\ST_2)$, we perform the following transitions:
\begin{center}
\includegraphics[scale=1]{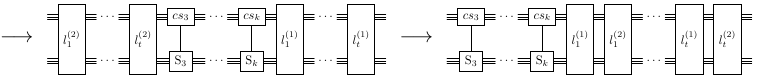}
\end{center}

\noindent
where we make use of the orthogonality between $l_{j}^{(1)}$ and $l_{j'}^{(2)}$. Performing the same transitions for $(cs_3,\ST_3)$, \dots, $(cs_k,\ST_k)$, we obtain the following arrangement:
\begin{center}
\includegraphics[scale=0.7]{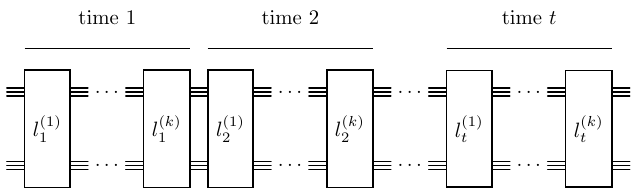}
\end{center}

\noindent
Given a certain $1\leq j\leq t$, we have that all controlled statements in $\cup_{i=1}^k l_j^{(i)}$  (that is, all controlled statements occurring in time $j$) are pairwise orthogonal. Therefore, we may rearrange their order according to their recursivity, and in doing so we may consider each time separately. For instance, in time $1$, let $\mathcal{L}\triangleq\cup_{i=1}^k l_1^{(i)}$, and let $\pn{proc$_1$},\, \pn{proc$_2$},\dots \pn{proc$_m$}$ denote procedures belonging to different recursion families. Then, we perform the following partition:
\begin{align*}
\mathcal{L}_0 &\triangleq \{ (cs,\ST)\in \mathcal{L}: \not\exists \proc \text{ such that }w_\PR^\proc(\ST^\proc)=1\text{ and }w_\PR^\proc(\ST)=1\}.\\
\mathcal{L}_i &\triangleq \{ (cs,\ST)\in \mathcal{L}: w_\PR^{\pn{proc$_i$}}(\ST)=1\},\quad i=1,\dots,m.
\end{align*}
By the definition of the sequential form of each controlled statement, we note that the partition is well-defined (e.g. there are no statements containing calls to more than one procedure). Therefore, we are able to rearrange $\mathcal{L}$ and perform calls to \optimize{} in the following way:
\begin{center}
\includegraphics[scale=0.8]{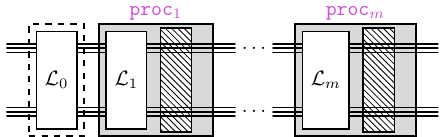}
\end{center}

Performing this operation on each time block and composing the circuits we obtain the rule given in Section~\ref{s:compilation}. This concludes the proof.
\end{proof}

We introduce the notion of rank of a procedure for use in the proof of Theorem~\ref{thm:new-compilation}.
\begin{definition}[Rank]
Set $\max(\emptyset)\triangleq 0$.
Given a program \PR, the rank of a procedure $\proc \in \PR{}$, denoted $rk_\PR(\proc)$, is defined as follows:

\scalebox{0.9}{
$
rk_\PR(\proc)\triangleq
\begin{cases}
0, &\text{ if }\neg (\exists \pn{proc'},\ \proc \succeq_\PR \pn{proc'}),\\
\max_{\proc \succeq_\PR \pn{proc'}} \{rk_\PR(\pn{proc'})  \}, &\text{ if }\exists \proc',\ \proc \succeq_\PR \pn{proc'} \wedge \neg (\proc \sim_\PR \proc),\\
1+ \max_{\proc\succ_\PR \pn{proc'} }\{rk_\PR(\pn{proc'}) \},&\text{ if }\proc \sim_\PR \proc.
\end{cases}
$}
\end{definition}

\thmnewcompilation*

\begin{proof}

The theorem can be shown by structural induction on the program body. All cases are straightforward except the one of the quantum control case. We make use of the following two facts regarding \pbp{} programs:
\begin{enumerate}[(a)]
\item The size of the circuit is directly proportional to its total number of \emph{unique} procedure calls (in the sense required for merging), and
\item a recursive procedure call results in $O(n)$ unique calls to procedures of the same rank. This is because unique calls may only differ on procedure name (of which there is a constant amount) or input size (for which there is a linear number of possibilites).
\end{enumerate}

%


Consider a quantum control statement $\ST=\tb{qcase }\q\tb{ of }\{0\to\ST_0,\,1\to\ST_1\}$ appearing in \optimize{} in the context of a recursive procedure \proc{}. By (a), the circuit size for $\ST_0$ and $\ST_1$ are proportional to the (total) number of procedure calls in each statement, separately.  We check that the number of ancillas created for $\ST$ is bounded by the maximum number of ancillas between $\ST_0$ and $\ST_1$.

 We proceed by induction on the rank $r$ of the procedure. The base case of $r=1$ is given by (b), and so we may consider $r>1$. For the inductive case, we consider two scenarios:
\begin{itemize}
\item  $w_\proc^\PR(\ST_0)=w_\proc^\PR(\ST_1)= 1$.
 In this case, both $\ST_0$ and $\ST_1$ are of rank $r$, and all their recursive procedure calls may be merged. Therefore, the asymptotic number of such calls is bounded between the maximum between $\ST_0$ and $\ST_1$ (consider that, if there is no overlap between the ancillas needed, their number is still bounded linearly). Applying the IH on the procedures of rank $r-1$ we obtain the desired result.
\item $w_\proc^\PR(\ST_0)=0$ and $w_\proc^\PR(\ST_1)= 1$. In this case, $\ST_0$ contains calls to procedures of rank $r'<r$ whereas $\ST_1$ contains calls to procedures of rank $r$. According to \optimize{}, statement $\ST_0$ is kept in the contextual circuit until no more statements are recursive relative to $\proc$. The statements of rank $r'$ which are present in $\ST_0$ are then merged with the equivalent procedures that were derived from $\ST_1$ and also added to $\contL$. The number of procedures of rank $r'$ is bounded asymptotically by the maximum between those in $\ST_0$ and $\ST_1$, therefore we obtain our result.\qedhere 
\end{itemize}
\end{proof}


\thmsoundnesscompleteness*
\begin{proof}
Since $\pbp \subsetneq \wf \cap \wi_{\leq 1}\subsetneq \pfoq{}$, with $\pfoq$ the language of~\cite{HPS23}, we have that $\llbracket \pbp{}\rrbracket \subseteq\llbracket  \wf \cap \wi_{\leq 1}\rrbracket \subseteq \llbracket \pfoq{}\rrbracket$ and, by $\pfoq{}$ soundness~\cite[Theorem 3]{HPS23}, it also holds that $\llbracket \pbp{}\rrbracket_{\geq \frac{2}{3}}\subseteq \fbqp{}$. Completeness can be proven by showing that $\pbp$ can simulate the function algebra in~\cite{Yamakami20}, known to be complete for quantum polynomial time. The proof can be done using the same construction as in~\cite[Theorem 5]{HPS23}.\end{proof}

\permutation*
\begin{proof}
Any permutation can be written as the composition of two sets of disjoint transpositions, and therefore any permutation can be performed in constant time, using two time steps~\cite{MN01}. To perform a \emph{controlled} permutation, it suffices to create $O(n)$ ancillas with the correct controlled state, which can be done in $O(\log n)$ depth with $O(n)$ gates.
\end{proof}

\section{Examples of Table~\ref{table:problems}}
\label{app:examples-table}

\begin{example}[Quantum Full Adder]
\label{ex:full-adder} Let $\textnormal{\texttt{ADD}}$ denote the following $\pbp{}$ program, where the following syntactic sugar:
\[\textnormal{TOF}(\q_1,\q_2,\q_3)\triangleq \tb{qcase }\q_1 \tb{ of }\{0\to\tb{skip};,1\to \textnormal{CNOT}(\q_2,\q_3)\}\]
encodes the Toffoli gate, i.e. the multi-controlled $NOT$ gate.
\begingroup
\addtolength{\jot}{0mm}
\begin{align*}
\ttt{1}& \quad \decl \pn{fullAdder}(\bar{\q})\{ \\
\ttt{2} & \qquad \tb{if }|\bar{\q}|>3 \tb{ then }
\codecomm{$\bar{\q}\el{1}=a$, $\bar{\q}\el{2}=b$, $\bar{\q}\el{-2}=\ket{0}$ and $\bar{\q}\el{-1}=c_\text{in}$}\\
\ttt{3}& \qquad \quad\textnormal{TOF}(\bar{\q}\el{1},\bar{\q}\el{2},\bar{\q}\el{-2})\\
\ttt{4}& \qquad \quad\textnormal{CNOT}(\bar{\q}\el{1},\bar{\q}\el{2})\\
\ttt{5}& \qquad \quad\textnormal{TOF}(\bar{\q}\el{2},\bar{\q}\el{-1},\bar{\q}\el{-2})\ \codecomm{$c_\text{out}=(a\cdot b)\oplus (c_\text{in} \cdot (a\oplus b)) $}\\
\ttt{6}& \qquad \quad\textnormal{CNOT}(\bar{\q}\el{2},\bar{\q}\el{-1})\ \codecomm{$s=a \oplus b\oplus c_\text{in}$}\\
\ttt{7}& \qquad \quad\textnormal{CNOT}(\bar{\q}\el{1},\bar{\q}\el{2})\\
\ttt{8}& \qquad \quad\tb{call }\pn{fullAdder}(\bar{\q}\ominus\el{1,2,-1});\\
\ttt{9}& \qquad \tb{else skip};\}\\
\ttt{10}&\quad::\ \tb{call }\pn{fullAdder}(\bar{\q});
\end{align*}
\endgroup
Given a carry-in bit $c_\text{in}$, we have that $\llbracket \textnormal{\texttt{ADD}}\rrbracket(\ket{a_n b_n \dots a_1 b_1 0^{n}c_\textnormal{in}})= \ket{a_n b_n \dots a_1 b_1 c_\text{out} s_1 \dots s_n},$ where $c_\text{out}$ encodes the carry-out bit and $s_i$ encodes the $i$-th sum bit. Given that \pn{fullAdder} performs one recursive call to an input containing three fewer qubits, we have that $\level_\textnormal{\texttt{ADD}}(n)=\lfloor \frac{n}{3}\rfloor + 1$.
\end{example}

\begin{example}[$k$-Chained Substrings]\label{ex:chained-substring}


Consider a program for detecting a substring $001$ occurring $k$ times in an input. For the case $k=1$, we define a program with procedures $\pn{a$_i$}$, $\pn{b$_i$}$, $\pn{c$_i$}$, $\pn{d$_i$}$ and $\pn{$\oplus$}$, with the graph:
\begin{center}
\includegraphics[scale=0.8]{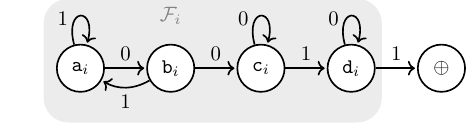}
\end{center}
The two outgoing edges of a node indicate the two branches of a \tb{qcase} statement in the corresponding procedure, controlled on the first input qubit. We denote by $\mathcal{F}_i$ the diagram containing nodes $\pn{a$_i$}$, $\pn{b$_i$}$, $\pn{c$_i$}$ and $\pn{d$_i$}$ and all edges between them. Procedures consist only of a \tb{qcase} statement, for the exception of $\pn{$\oplus$}$:
\begin{align*}
&\tb{decl }\pn{a$_i$}(\bar{\q})\{\,\tb{qcase }\bar{\q}\el{1}\tb{ of }\{0\to\tb{call }\pn{b$_i$}(\bar{\q}\ominus\el{1});,\,1\to\tb{call }\pn{a$_i$}(\bar{\q}\ominus\el{1});\, \}\,\},\\
&\tb{decl }\pn{b$_i$}(\bar{\q})\{\,\tb{qcase }\bar{\q}\el{1}\tb{ of }\{0\to\tb{call }\pn{c$_i$}(\bar{\q}\ominus\el{1});,\,1\to\tb{call }\pn{b$_i$}(\bar{\q}\ominus\el{1});\, \}\,\},\\
&\tb{decl }\pn{c$_i$}(\bar{\q})\{\,\tb{qcase }\bar{\q}\el{1}\tb{ of }\{0\to\tb{call }\pn{c$_i$}(\bar{\q}\ominus\el{1});,\,1\to\tb{call }\pn{d$_i$}(\bar{\q}\ominus\el{1});\, \}\,\},\\
&\tb{decl }\pn{d$_i$}(\bar{\q})\{\,\tb{qcase }\bar{\q}\el{1}\tb{ of }\{0\to\tb{call }\pn{d$_i$}(\bar{\q}\ominus\el{1});,\,1\to\tb{call }\pn{$\oplus$}(\bar{\q}\ominus\el{1});\, \}\,\},\\
&\tb{decl }\pn{$\oplus$}(\bar{\q})\{\,\bar{\q}\el{-1}\asg \text{NOT};\, \}
\end{align*}
The program body is a call to procedure \pn{a$_i$} on input $\bar{\q}$, which results in a program that performs the transformation $\ket{\bar{x}y}\mapsto \ket{\bar{x}(y\oplus b)}$ where $b\in\{0,1\}$ is 1 if and only if $\bar{x}$ contains at least $1$ instance of $001$ as a substring.

For a general $k$, we consider the program $\PR_k$ defined by
\begin{center}
\includegraphics[scale=0.8]{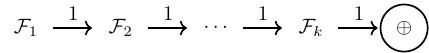}
\end{center}
where an arrow from $\mathcal{F}_i$ to $\mathcal{F}_{i+1}$ indicates an arrow from \pn{d$_{i}$} to \pn{a$_{i+1}$}. The final program then consits of a call to procedure \pn{a$_1$} on input $\bar{\q}$. It is easy to check that $\PR_k\in\pbp{}$, and that $\level_{\PR_k}(n)=n$, for any $k$.
\end{example}

\begin{example} \label{ex:sumr} Let $\text{Sum}(r)$ be the decision problem of checking if an input bitstring $x\in\{0,1\}^n$ satisfies $\sum_{i=1}^n x_i = r$. For instance, if $r=3$, we may define a program \textnormal{\texttt{SUM\_3}} as:
\begin{align*}
\ttt{1} &\quad \tb{decl }\pn{zero}(\bar{\q})\{\,\tb{qcase }\bar{\q}\el{1}\tb{ of }\{0\to\tb{call }\pn{zero}(\bar{\q}\ominus\el{1});,\,1\to\tb{call }\pn{one}(\bar{\q}\ominus\el{1});\, \}\,\},\\
\ttt{2} &\quad\tb{decl } \pn{one}(\bar{\q})\{\,\tb{qcase }\bar{\q}\el{1}\tb{ of }\{0\to\tb{call } \pn{one}(\bar{\q}\ominus\el{1});,\,1\to\tb{call }\pn{two}(\bar{\q}\ominus\el{1});\, \}\,\},\\
\ttt{3} &\quad\tb{decl } \pn{two}(\bar{\q})\{\,\tb{qcase }\bar{\q}\el{1}\tb{ of }\{0\to\tb{call } \pn{two}(\bar{\q}\ominus\el{1});,\,1\to\tb{call }\pn{three}(\bar{\q}\ominus\el{1});\, \}\,\},\\
\ttt{4} &\quad\tb{decl } \pn{three}(\bar{\q})\{\\
\ttt{5} &\quad\quad\tb{if }|\bar{\q}|=1\tb{ then }\\
\ttt{6}&\qquad \qquad\tb{call }\pn{$\oplus$}(\bar{\q});\\
\ttt{7}&\qquad\tb{else}\\
\ttt{8}&\qquad\qquad\tb{qcase }\bar{\q}\el{1}\tb{ of }\{0\to\tb{call } \pn{three}(\bar{\q}\ominus\el{1});,\,1\to\tb{skip};\,\}\,\},\\
\ttt{9} &\quad\tb{decl }\pn{$\oplus$}(\bar{\q})\{\,\bar{\q}\el{-1}\asg \textnormal{NOT};\, \}\\
\ttt{10}&\quad::\ \tb{call }\pn{zero}(\bar{\q});
\end{align*}
We have that $\level_\textnormal{\texttt{SUM\_3}}(n)=n$.
\end{example}

\end{document}